\documentclass[acmsmall]{acmart}
\usepackage{soul}
\usepackage{tikz, tikz-qtree}

\newtheorem{contribution}{Main Contribution}
\newtheorem{motivatingexample}{Motivating Example}
\newtheorem{remark}{Remark}

\newcommand{\stricthom}{\begin{array}{l}\to\\[-0.7em] \nleftarrow\end{array}}

\newcommand{\skeleton}{\operatorname{skeleton}}

\title{Conjunctive Queries: Unique Characterizations and Exact Learnability}

\acmJournal{TODS}



\begin{document}

\author{Balder ten Cate}
\orcid{https://orcid.org/0000-0002-2538-5846}
\affiliation{
  \institution{Universiteit van Amsterdam}
  \department{ILLC}
  \city{Amsterdam}
  \country{The Netherlands}}
\email{b.d.tencate@uva.nl}
\author{Victor Dalmau}
\orcid{https://orcid.org/0000-0002-9365-7372}
\affiliation{
  \institution{Universitat Pompeu Fabra}
  \department{Department of Information and Communication Technologies}
  \city{Barcelona}
  \country{Spain}}
\email{victor.dalmau@upf.edu}

\begin{abstract}
We answer the question of which conjunctive queries are uniquely characterized by polynomially many positive and negative examples, and how to construct such examples efficiently.
As a consequence, we obtain a new efficient exact learning algorithm for a class of conjunctive queries. At the core of our contributions lie two new polynomial-time algorithms for constructing frontiers in the homomorphism lattice of finite structures. We also discuss implications for the unique characterizability and learnability of schema mappings and of description logic concepts. 
\end{abstract}

\maketitle

\section{Introduction}

Conjunctive queries (CQs) are an extensively studied database query language and fragment of first-order logic. 
They correspond precisely to 
Datalog programs with a single non-recursive rule.
In this paper, we study two problems related to CQs.
The first problem is concerned with the existence and constructability of unique characterizations. 
\emph{For which CQs $q$ is it the case that $q$ can be characterized (up to logical equivalence) by its behavior on a small set of data examples? And, when such a set of data examples exists, can it be constructed efficiently?}
The second problem pertains to \emph{exact learnability} of CQs in an interactive
setting where the learner has access to a ``membership oracle'' that, given any database instance and a tuple of values, answers whether the tuple belongs to the answer of the goal CQ (that is, the hidden CQ that the learner is trying to learn). We can 
think of the membership oracle as a black-box, compiled version of the goal query, which the 
learner can execute on any number of examples. The task of the learner, then, is to reverse engineer the query based on the observed behavior.

\begin{figure}
\begin{center}
\begin{tikzpicture}[level distance=36pt,sibling distance=4pt]
\Tree [.\node(top){$\top$} ; \edge[draw=none];
[.\text{$A$}
\edge[-] ; [.\node(f1){$F_1$}; ]
\edge[draw=none] ; [.\text{\dots} \edge[draw=none] ; [.{} \edge[draw=none]; [.\node(bot){$\bot$}; ] ]]
\edge[-] ; [.\node(fn){$F_n$}; ] ]
]
\draw (top) edge[out=220,in=140,looseness=1.2] (bot.west);
\draw (top) edge[out=-40,in=40,looseness=1.2] (bot.east);
\draw (f1) edge[out=250,in=130] (bot.west);
\draw (fn) edge[out=290,in=50] (bot.east);
\draw[dashed] (-1.2,-2) rectangle (1.2,-2.8);
\end{tikzpicture}
\end{center}
\caption{A frontier in the homomorphism lattice of structures}
\label{fig:frontier}
\end{figure}

Note that these two problems (unique characterizability and 
exact learnability) are closely related to each other: a learner can  
identify the goal query with certainty, only when the set of examples that it has seen so far constitutes a unique characterization of the goal query. In other
words, unique characterizability (by a polynomially large set of examples)
is necessary, but not sufficient, for (polynomial-time) exact learnability
with membership queries.

\begin{motivatingexample}
This example, although stylized and described at a high level, 
aims to convey a use case that motivated the present work. 
The \emph{Google Knowledge Graph} is a large database of entities and facts, 
gathered from a variety of sources. It is used to enhance the search 
engine's results for queries such as
``where was Barack Obama born'' with factual information in the form of 
\emph{knowledge panels} \cite{KGBlog2020}.
When a query triggers a specific knowledge panel, this may be the result of different
triggering and fulfillment mechanisms, each of which may involve a combination of
structured queries to the knowledge graph, hard-coded business logic (in a Turing-complete language), and machine learned models.
This makes it difficult to understand interactions between knowledge panels (e.g., whether the two knowledge panels are equivalent or one is subsumed by the other, in terms of 
content and triggering).
If a declarative specification of (an approximation of) the triggering and fulfillment logic for a knowledge panel can be constructed programmatically, specified in a sufficiently restrictive formalism such as Datalog rules, this provides an avenue to the above, and other relevant static analysis tasks. The 
\emph{efficient exact learnability with membership queries} that we study in this paper, can be viewed as an idealized form of such a programmatic approach, where the membership oracle is the existing, 
black box, implementation of the knowledge panel, and the learning algorithm aims to produce a CQ that exactly captures it.
\end{motivatingexample}

The above example provides a motivation for studying efficient exact learnability of CQs, and hence, for studying unique characterizability. However, we would like to emphasize that unique characterizations are of independent interest, outside the context of exact learning algorithms. Indeed, uniquely characterizing examples can be used, for instance, for elementary query engine debugging, and query visualization and explanation.

As it turns out, the above problems about CQs are intimately linked to  fundamental properties of the homomorphism lattice of finite structures. In 
particular, the existence of a unique characterization for a CQ can be
reduced to the existence of a \emph{frontier} in the homomorphism lattice for an associated structure $A$, where, by a ``frontier'' for $A$, we mean a finite set of structures $F_1, \ldots, F_n$ that
cover precisely the set of structures homomorphically strictly weaker than $A$, that is, such that $\{B\mid B\to A \text{ and } A\not\to B\} = \bigcup_i \{B\mid B\to F_i\}$ (cf.~Figure~\ref{fig:frontier}). 

Known results \cite{FoniokNT08,AlexeCKT2011} imply that not every finite structure has such a frontier, and, moreover, a finite structure has a frontier if and only if the structure (modulo homomorphic equivalence) satisfies a structural property called \emph{c-acyclicity}. These known results, however, are based on exponential constructions, and no polynomial algorithms for constructing frontiers were previously known.

\begin{contribution}[Polynomial-time algorithms for constructing frontiers]
We show that, for c-acyclic structures, a frontier can in fact be computed \emph{in polynomial time}. 
More specifically, we present two polynomial-time algorithms.
The first algorithm takes any c-acyclic structure and produces a frontier consisting of structures that are themselves not necessarily c-acyclic (Sect.~\ref{sec:frontier-cacyclic}). 
The second algorithm applies to a more restricted class of acyclic structures but yields a frontier consisting entirely of structures belonging to the same
class, that is, the class of structures in question is frontier-closed (Sect.~\ref{sec:frontier-closed}).
\end{contribution}

We use these  to obtain new results on the existence and efficient constructability of unique characterizations for CQs:

\begin{contribution}[Polynomial Unique Characterizations for Conjunctive Queries]
We show that a CQ is uniquely characterizable by polynomially many examples, precisely if (modulo logical equivalence) it is 
 c-acyclic. Furthermore, for c-acyclic CQs, a uniquely characterizing set of examples can be constructed in polynomial time. 
In the special case of acyclic and c-connected CQs, a uniquely characterizing set of examples can be constructed consisting entirely of
queries from the same class (Sect.~\ref{sec:characterizations}).
\end{contribution}

Using the above results as a stepping stone, we obtain a polynomial-time exact learning algorithm for the class of c-acyclic CQs.

\begin{contribution}[Polynomial-Time Learnability with Membership Queries]
We show that c-acyclic CQs are efficiently exactly learnable in Angluin's model of exact learnability with membership 
queries~\cite{Angluin88} (Sect.~\ref{sec:learning}).
\end{contribution}

The restriction to c-acyclic CQs in this learnability result is natural, given that,
as we mentioned above, exact learnability with membership queries requires the existence of a finite uniquely characterizing set of examples. Note however, that our results do
not preclude the possibility that there exist larger efficiently exactly learnable classes of CQs: even if a class $C$ includes non-c-acyclic CQs, it may still be
possible for every CQ $q\in C$ to be uniquely characterizable \emph{within the class $C$}. 

We mainly focus on positive and negative examples in this paper. Another natural type of data example is a pair $(I,R)$ where $I$ is an input instance and $R$ is the entire relation that is computed by the query on $I$. We discuss this in Section~\ref{sec:input-output}, where we point out that all our results on characterizability and learnability remain true also when considering such data examples.

Finally, although our primary interest is in conjunctive queries, we show that our results also have implications for \emph{schema mappings} and \emph{description logic concepts}:

\begin{contribution}[Schema Mappings and Description Logic Concepts]
As a further corollary to the above, in Sect.~\ref{sec:applications}, we obtain a number of results regarding the existence of polynomial unique characterizations, as well as exact learnability, for LAV (``Local-As-View'') schema mappings and for description logic concepts for the lightweight description logic $\mathcal{ELI}$~\cite{Baader2017:introduction}.
\end{contribution}

\subsection*{Related Work}
Unique characterizations for CQs were first studied 
in~\cite{MannilaR86} in the context of automatic test data generation.
More precisely, the authors propose the concept of an ``adequate test case'', which 
is a database instance that can be used to distinguish a given CQ
from all other, non-equivalent CQs from a given class. In our
terminology, this corresponds to a uniquely characterizing input-output
example (cf.~Section~\ref{sec:input-output}). A positive result was 
obtained in~\cite{MannilaR86} for restricted classes of self-join-free
CQs, by establishing a relationship to Armstrong databases~\cite{Armstrong1974}. In~\cite{AlexeCKT2011},
the authors study unique characterizations for 
various classes of \emph{schema mappings}; we will make use of some of the technical results from~\cite{AlexeCKT2011}, and in Section~\ref{sec:applications}, we will discuss 
an application of our results to LAV schema mappings.
In~\cite{Staworko15:characterizing}, the authors study unique characterizability
for XML twig queries.

Related work on \emph{learning} CQs will be discussed in Section~\ref{sec:learning}.

An earlier, extended-abstract version of this paper was published in~\cite{tCD2021:conjunctive}. The present paper extends this conference version
with additional results. In particular, Theorem~\ref{thm:frontier-closed} was
shown in~\cite{tCD2021:conjunctive} only for the case of $k=1$ (i.e., for unary CQs), and
is generalized here to all $k\geq 1$. Furthermore, the treatment of 
input-output examples in Section~\ref{sec:input-output} has been added.

\subsection*{Outline} 
Section~\ref{sec:preliminaries} reviews basic facts and definitions.
In Section~\ref{sec:frontiers}, we present our two new polynomial-time algorithms for constructing frontiers
for finite structures with distinguished elements. We also review a result by \cite{NesetrilOssona2008}, which implies the 
existence of (not necessarily polynomially computable) frontiers w.r.t.~classes of 
structures of bounded expansion.
In Section~\ref{sec:characterizations}, we apply these  algorithms to show that a CQ is uniquely characterizable by polynomially many examples, precisely if (modulo logical equivalence) it is 
 c-acyclic. Furthermore, for c-acyclic CQs, a uniquely characterizing set of examples can be constructed in polynomial time. 
In the special case of unary, acyclic, connected CQs, a uniquely characterizing set of examples can be constructed consisting entirely of
queries from the same class. 
In Section~\ref{sec:learning}, we further build on these results, and we study the exact learnability of 
CQs. 
In Section~\ref{sec:input-output}, we consider another type of data examples, namely \emph{input-output examples}.
Section~\ref{sec:applications}, finally, presents applications to  
schema mappings and description logic concepts.


\section{Preliminaries}
\label{sec:preliminaries}


\subsection*{Schemas, Structures, Homomorphisms, Cores}
A \emph{schema} (or, relational signature) is a finite set of relation symbols $\mathcal{S}=\{R_1, \ldots, R_n\}$, where each relation $R_i$ as an associated arity $\text{arity}(R_i) \geq 1$. For $k\geq 0$, by a \emph{structure over $\mathcal{S}$ with $k$ distinguished elements} we will mean a tuple $(A, a_1, \ldots, a_k)$, where $A=(\text{dom}(A), R^A_1, \ldots, R^A_n)$ is a finite structure (in the traditional, model-theoretic sense) over the schema $\mathcal{S}$, and $a_1, \ldots, a_k$ are elements of the domain of $A$. Note that all structures, in this paper, are assumed to be finite, and we will drop the adjective ``finite''. By a \emph{fact} of a structure $A$ we mean an expression of the form $R(a_1, \ldots, a_n)$ where the tuple $(a_1, \ldots, a_n)$ belongs to the relation $R$ in $A$. Given two structures $(A,\textbf{a})$ and $(B,\textbf{b})$ over the same schema, where $\textbf{a}=a_1, \ldots a_k$ and $\textbf{b}=b_1, \ldots, b_k$, a 
\emph{homomorphism} $h: (A,\textbf{a})\to (B,\textbf{b})$ is a map $h$ from the domain of $A$ to the domain of $B$, such that $h$ preserves all facts
(i.e., for each fact $R(a_1, \ldots, a_n)$ of $A$, $R(h(a_1), \ldots, h(a_n))$ is a fact of $B$),
and such that $h(a_i)=b_i$ for $i=1\ldots k$. When such a homomorphism exists, we will also say that $(A,\textbf{a})$ ``homomorphically maps to'' $(B,\textbf{b})$ and we will write
$(A,\textbf{a})\to (B,\textbf{b})$. We say that $(A,\textbf{a})$ and $(B,\textbf{b})$ are \emph{homomorphically equivalent} if
$(A,\textbf{a})\to (B,\textbf{b})$ and $(B,\textbf{b})\to (A,\textbf{a})$.
We occasionally write $A\stricthom B$ to say that $A\to B$ and $B\not\to A$.

A structure is said to be a \emph{core} if there is no homomorphism from the structure in question to a proper substructure~\cite{HN92}. It is known~\cite{HN92} that every structure $(A,\textbf{a})$ has a substructure to which it is homomorphically equivalent and that is a core. This substructure, moreover, is unique up to isomorphism, and it is known as \emph{the core of $(A, \textbf{a})$}. 

We will make use of the following technical lemma in several places later on in this paper:

\begin{lemma}\label{lem:core-injective}
Let $(A,\textbf{a})$ be a core structure, and let $h:(B,\textbf{b})\to (A,\textbf{a})$.
If there is a homomorphism $h':(A,\textbf{a})\to (B,\textbf{b})$, then $h'$ must be 
injective, and moreover, in this case, 
there is such $h'$ with the additional property that the composition of $h$ with $h'$ is the identity map on $A$.
\end{lemma}

\begin{proof}
The composition of $h$ with $h'$ is an endomorphism of $(A,\textbf{a})$ (that is, 
a homomorphism from the structure to itself).  It is a well-known property of cores
that every endomorphism is an automorphism (that is, an isomorphism from the structure to itself). Therefore, $h'$ must be injective. Furthermore, by composing $h'$ with the inverse of the automorphism, we ensure that its composition with $h$ is the identity function.
\end{proof}

\subsection*{Fact Graph, FG-Connectedness, FG-Disjoint Union}
The \emph{fact graph} of a structure $(A,\textbf{a})$ is the undirected graph whose nodes are the facts of $A$, and such that there is an edge between two distinct facts if they share a non-distinguished element, i.e., there exists an element $b$ of the domain of $A$ that is distinct from the distinguished elements $\textbf{a}$, such that $b$ occurs in both facts. 
We say that $(A,\textbf{a})$ is \emph{fg-connected} if the fact graph is connected. 
A \emph{fg-connected component} of $(A,\textbf{a})$ is a 
maximal fg-connected substructure $(A',\textbf{a})$ of $(A,\textbf{a})$.
If $(A_1,\textbf{a})$ and $(A_2,\textbf{a})$ are structures with the same 
distinguished elements, and whose domains are otherwise (except for these distinguished elements) disjoint, then the union $(A_1\cup A_2, \textbf{a})$ of these two structures will be called a \emph{fg-disjoint union} and 
will be denoted as $(A_1,\textbf{a})\uplus (A_2,\textbf{a})$. The same
construction naturally extends to finite sets of structures. 
It is easy to see that every structure $(A,\textbf{a})$ is equal to the \text{fg-disjoint union} of its fg-connected components. See also~\cite{Fagin2008towards,tencate2009laconic}, where
fg-connected components are called \emph{fact blocks}.

\subsection*{Direct Product, Homomorphism Lattice} 
Given two structures $(A,\textbf{a})$ and $(B,\textbf{b})$ over the same schema, where $\textbf{a}=a_1, \ldots a_k$ and $\textbf{b}=b_1, \ldots, b_k$, the \emph{direct product} $(A,\textbf{a})\times (B,\textbf{b})$ is defined, as usual, as $(A\times B, \langle a_1, b_1\rangle, \ldots, \langle a_k, b_k\rangle)$, where the domain of $A\times B$ is the Cartesian product of the domains of $A$ and $B$, and where the
facts of $A\times B$ are all facts $R(\langle c_1,d_1\rangle, \ldots, \langle c_n, d_n\rangle)$ for which it holds that
$R(c_1, \ldots, c_n)$ is a fact of $A$ and $R(d_1, \ldots, d_n)$ is a fact of $B$.
The direct product of a finite collection of structures is defined analogously.

For a fixed schema $\mathcal{S}$ and $k\geq 0$, the collection of homomorphic-equivalence classes of structures over $\mathcal{S}$ with $k$ distinguished elements, ordered by homomorphism, forms a lattice. Specifically,  the above
\emph{direct product} operation is a meet operation in the lattice-theoretic sense:
$(A,\textbf{a})\times (B,\textbf{b})$ homomorphically maps to both $(A,\textbf{a})$ and $(B,\textbf{b})$, and a structure $(C,\textbf{c})$ homomorphically maps to $(A,\textbf{a})\times (B,\textbf{b})$ if and only if it homomorphically maps to both $(A,\textbf{a})$ and $(B,\textbf{b})$. The join operation of the lattice is a little more tedious to define,
and we only sketch it here, as it is not used in the remainder of the paper. 
For a structure $(A,\textbf{a})$ with $\textbf{a}=a_1, \ldots, a_k$, by the 
\emph{isomorphism type of the distinguished elements} we will mean the 
equivalence relation over $\{1, \ldots, k\}$ induced by the tuple $a_1, \ldots, a_k$.
When two structures have the same isomorphism type of distinguished elements, their join is simply the fg-disjoint union as defined earlier. In the general case, one must first compute the smallest equivalence relation over $\{1, \ldots, k\}$ that refines the
isomorphism type of distinguished elements of both structures, and factor both structures through this equivalence relation, before taking their fg-disjoint union.

For structures without distinguished elements, this lattice has been studied extensively (cf.~for instance~\cite{HellNesetril2004,nesetril2012sparsity}). The above exposition shows how to lift some of the fundamental constructions to structures with distinguished elements. As we will see, it will be important in much of  this paper to consider structures with distinguished elements, as 
these distinguished elements, intuitively, 
correspond to the free variables of a CQ.

\subsection*{Incidence Graph, Acyclicity, C-Acyclicity}
Given a structure $(A,\textbf{a})$, 
the \emph{incidence graph} of $A$ is the bipartite  multi-graph containing all elements of the domain of $A$ as well as all facts of $A$, and an 
edge $(a,f)$ whenever $a$ is an element and $f$ is a fact in which $a$ occurs. Whenever an element $a$ occurs more than once in the same fact $f$, the incidence graph contains a distinct edge for every occurrence of $a$ in $f$. We will call a structure
$(A,\textbf{a})$ \emph{acyclic} (also known as \emph{Berge-acyclic}~\cite{Fagin83:acyclicity}) if the incidence graph of $A$ is acyclic;
$(A,\textbf{a})$ 
 is said
to be \emph{c-acyclic} if every cycle in its incidence graph contains at least one distinguished element, i.e., at least one element in $\textbf{a}$. 
In particular, acyclicity implies that no element occurs twice in the same fact, and c-acyclicity implies that no non-distinguished element occurs twice in the same fact.
In the case without distinguished elements, c-acyclicity is equivalent to acyclicity.
The concept of c-acyclicity was first introduced in~\cite{AlexeCKT2011} in the study of unique characterizability of GAV schema mappings (cf.~Section~\ref{sec:applications} for more details).
A straightforward dynamic-programming argument shows~\cite{Dalmau02:constraint}:

\begin{proposition}\label{prop:cacyclic-core}
For c-acyclic structures $(A, \textbf{a})$ and $(B,\textbf{b})$ (over the same schema and with the same number of distinguished elements), we can test in polynomial time whether $(A,\textbf{a})\to(B,\textbf{b})$.
The core of a c-acyclic structure can be computed in polynomial time.
\end{proposition}

\subsection*{C-Connectedness}
We say that a structure $(A,\textbf{a})$ is \emph{c-connected} if 
every connected component of its incidence graph contains at least one distinguished element. 
Note that this condition is only meaningful for structures with at least one distinguished element, and that it
differs subtly from the condition of fg-connectedness we defined above. For example, the structure consisting of the facts $R(a_1, a_2)$ and $S(a_2, a_1)$ with distinguished elements $a_1, a_2$, is c-connected but is \emph{not} fg-connected. 
For any structure $(A,\textbf{a})$, we denote by
$(A, \textbf{a})^{\textrm{reach}}$ the (unique) maximal c-connected substructure, that is, the substructure containing everything reachable from the distinguished elements.

\begin{proposition}\label{prop:reach}
If $(A,\textbf{a})$ is c-connected, then 
$(A,\textbf{a})\to (B, \textbf{b})^{\textrm{reach}}$ iff 
$(A,\textbf{a})\to (B, \textbf{b})$.
\end{proposition}

\subsection*{Conjunctive Queries}
Let $k\geq 0$.
A $k$-ary \emph{conjunctive query (CQ)} $q$ over a schema $\mathcal{S}$ is an expression of the form ~
$ q(\textbf{x}) \text{ :- } \alpha_1\land\cdots\land\alpha_n $ ~
where $\textbf{x}=x_1, \ldots, x_k$ is a sequence of variables, and where each $\alpha_i$ is an atomic formula using a relation from $\mathcal{S}$ and using variables as arguments only. Note that $\alpha_i$ may use variables from $\textbf{x}$ as well as other variables.
In addition, it is 
required that each variable in $\textbf{x}$ occurs in at least one
conjunct $\alpha_i$. This requirement is referred to as the \emph{safety} condition.

Note that, for simplicity, this definition of CQ does not allow the use of constants. 
Many of the results in this paper, however,
can be extended in a straightforward way to CQs with a fixed finite number of constants (which can be simulated using additional free variables).

If $A$ is a structure over the same schema as $q$, we denote by $q(A)$ the
set of all $k$-tuples of values that satisfy the query $q$ in $A$. 
We write $q\subseteq q'$ if $q$ and $q'$ are queries over the same schema, and of the 
same arity, and $q(A)\subseteq q'(A)$ holds for all structures $A$. We say that
$q$ and $q'$ are \emph{logically equivalent} if $q\subseteq q'$ and $q'\subseteq q$ both hold. We refer to any textbook on database theory for a more detailed exposition of the 
semantics of CQs, and we will restrict ourselves to giving an equivalent presentation of 
the semantics of CQs through canonical structures and the Chandra-Merlin theorem.

There is a well-known correspondence between $k$-ary CQs over a schema 
$\mathcal{S}$ and structures over $\mathcal{S}$ with $k$ distinguished elements. 
In one direction, 
we can associate to each $k$-ary CQ $q(\textbf{x})$ over the schema $\mathcal{S}$ a corresponding structure over $\mathcal{S}$ with $k$ distinguished elements, namely $\widehat{q} = (A_q,\textbf{x})$, where the domain of $A_q$ is the set of variables
occurring in $q$, and the facts of $A_q$ are the conjuncts of $q$. We will call this structure $\widehat{q}$ the \emph{canonical structure} of $q$. Note that
every distinguished element of $\widehat{q}$ occurs in at least one fact, as follows
from the safety condition of CQs.
Conversely, consider any structure $(A,\textbf{a})$, with $\textbf{a}=a_1, \ldots, a_k$,  such that every distinguished element
$a_i$ occurs in at least one fact of $A$. We can associate to $(A,\textbf{a})$ a
$k$-ary \emph{canonical CQ}, namely the CQ 
 that has a variable $x_a$ for every value $a$ in the domain of $A$ occurring in at least one fact, and a conjunct for every fact of $A$.

 By the classic \emph{Chandra-Merlin Theorem}~\cite{CM77}, 
 a tuple $\textbf{a}$ belongs to $q(A)$ if and only if there
 is a homomorphism from $\widehat{q}$ to $(A,\textbf{a})$; and 
 $q\subseteq q'$ holds if and only if there is a homomorphism from $\widehat{q'}$ to $\widehat{q}$. Finally, $q$ and $q'$ are logically equivalent if and only if 
 $\widehat{q}$ and $\widehat{q'}$ are homomorphically equivalent.

\subsection*{Exact Learning Models, Conjunctive Queries as a Concept Class}
Informally, an \emph{exact learning algorithm} is an algorithm that
identifies an unknown goal concept by asking a number of queries about
it. The queries are answered by an oracle that has access to the goal
concept. This model of learning was introduced by Dana Angluin, cf.~\cite{Angluin88}.
In this paper, we consider the two most extensively studied
kinds of oracle queries: \emph{membership queries}  and \emph{equivalence queries}. 
We will first review basic notions from computational
learning theory, such as the notion of a \emph{concept}, and then
explain what it means for a concept class to be \emph{efficiently
  exactly learnable with membership and/or equivalence queries}. 
  
Let $X$ be a (possibly infinite) set of {\em examples}. 
A \emph{concept over $X$} is a function $c:X\to \{0,1\}$, and a \emph{concept class} ${\mathcal C}$ is a collection of such concepts.
We say that $x\in X$ is a \emph{positive example} for a concept $c$ if $c(x)=1$, 
and that $x$ is a \emph{negative example} for $c$ if $c(x)=0$.

Conjunctive queries (over a fixed schema $\mathcal{S}$ and with a fixed arity $k$) are a particular example of such a concept class, where the example space is the class of all
structures over $\mathcal{S}$ with $k$ distinct elements, and where an example $(A,\textbf{a})$ is labeled as positive if the tuple $\textbf{a}$ belongs to $q(A)$, and negative otherwise.

It is always assumed that concepts are specified using some representation system so that one can speak of the length of the specification of a concept. More formally, a \emph{representation system for $\mathcal{C}$} is a string language $\mathcal{L}$ over some finite alphabet, together with a surjective function $r:\mathcal{L}\to\mathcal{C}$. By the 
 \emph{size} of a concept $c\in\mathcal{C}$, we will mean the length of the smallest representation.
Similarly, we assume a representation system, with a corresponding notion of length, for the examples in $X$.
When there is no risk of confusion, we may conflate concepts (and examples) with their representations.

Specifically, for us, when it comes to \emph{structures}, any natural choice of representation will do; we only assume that the length of the specification of a structure (for a fixed schema) is polynomial in the domain size, the number of facts and the number of distinguished elements. Likewise for \emph{CQs},  we   assume that the length of the representation of a CQ is polynomial in that of its canonical structure.

\newcommand{\alg}{\texttt{alg}}
For every concept $c$, we denote by $\text{\rm MEM}_c$ the
\emph{membership oracle} for $c$, that is,  the oracle that takes as
input an example $x$ and returns its label, $c(x)$, according to
$c$. Similarly, for every concept $c\in \mathcal{C}$, we denote by
$\text{\rm EQ}_c$, the \emph{equivalence oracle} for $c$, that is,
the oracle that takes as input the representation of a concept $h$ and
returns ``yes'', if $h=c$, or returns a counterexample $x$ otherwise
 (that is, an example $x$ such that $h(x)\neq
c(x)$).
An \emph{exact learning algorithm with membership and/or equivalence queries} for a concept class $C$ is an algorithm $\alg$ that takes no input but has access to the membership oracle $\text{\rm MEM}_c$ and/or equivalence oracle $\text{\rm EQ}_c$ for some concept
$c\in C$, which will be called the \emph{goal concept}. Importantly, while the algoritm
may interact with the oracle(s), it does not know which concept $c\in C$ is the goal concept.
\footnote{It is common in the learning theory literature to assume that the learning algorithm is given 
an upper bound on
the size of the goal concept as input. However, it turns out that
such an assumption is not needed for any of our positive results concerning learnability.}
 Intuitively, the algorithm $\alg$ must determine $c$ by asking oracle queries. More precisely, for every choice of $c\in C$,
 $\alg$ must terminate after a finite amount of time, and output (some representation of) the goal concept $c$. This notion was introduced by Angluin~\cite{Angluin88}, who also introduced  the notion of a \emph{polynomial-time} exact learning algorithm.
 We say that an exact learning algorithm $\alg$ with membership and/or equivalence queries
 \emph{runs in polynomial time} if there exists a two-variable polynomial $p(n,m)$ such that
 at any point  during the run of the algorithm, the time used by $\alg$ up to  that point (counting one step per oracle call) is bounded by $p(n,m)$, where $n$ is the size of the goal concept and $m$ the size of the largest counterexample returned by calls to the equivalence oracle up to that point in the run ($m=0$ if no equivalence queries have been used).
 A concept class ${\mathcal C}$ is \emph{efficiently exactly learnable with membership and/or equivalence queries} if there is an exact learning algorithm with membership and/or equivalence queries for $\mathcal{C}$ that runs in polynomial time.

There is a delicate issue about this notion of polynomial time that we now discuss.
One might be tempted to relax the previous definition by requiring merely that the total running time is bounded by $p(n,m)$. However, this change in the definition would give rise to a {\em wrong} notion of a  polynomial-time algorithms in this context by way of a loophole in the definition. Indeed, under this change, one could design a learning algorithm that, in a first stage, identifies the goal hypothesis by (expensive) exhaustive search and that, once this is achieved, forces ---by asking equivalence queries with hypotheses that are appropriate modifications of the goal concept--- the equivalence oracle to return large counterexamples that would make up for the time spent during the exhaustive search phase.

\section{Frontiers in the homomorphism lattice of structures}
\label{sec:frontiers}

In this section, we define frontiers, as well as the relation notions of
gap pairs and (restricted) homomorphism dualities, and we will discuss their relationships.
We present two 
polynomial-time methods for constructing frontiers. 

For the applications in the next sections, it is important to 
consider structures with distinguished elements. 
These distinguished elements, intuitively, 
correspond to the free variables of a CQ. 
Specifically, Proposition~\ref{prop:frontiers-characterizations} in Section~\ref{sec:characterizations} will link unique
characterizations for $k$-ary CQs to frontiers for structures with $k$ distinguished elements.
For this reason, all the results in this section are stated for structures with distinguished 
elements.

\begin{definition}
Fix  a schema and $k\geq 0$, and let $\mathcal{C}$ be a class of structures with $k$ distinguished elements and let $(A,\textbf{a})$ be a structure with $k$ distinguished elements. 
A \emph{frontier} for $(A,\textbf{a})$ w.r.t.~$\mathcal{C}$,
is a finite set of structures $F$ such that
\begin{enumerate}
\item $(B,\textbf{b})\to(A,\textbf{a})$ for all $(B,\textbf{b})\in F$.
\item $(A,\textbf{a})\not\to(B,\textbf{b})$ for all $(B,\textbf{b})\in F$.
\item For all $(C,\textbf{c})\in\mathcal{C}$ with $(C,\textbf{c})\to (A,\textbf{a})$ and $(A,\textbf{a})\not\to(C,\textbf{c})$, we have that $(C,\textbf{c})\to (B,\textbf{b})$ for some $(B,\textbf{b})\in F$.
\end{enumerate}
\end{definition}

See Figure~\ref{fig:frontier} for a graphical depiction of a frontier.

The notion of a frontier is closely related to that of a gap pair. 
While gap pairs will not play an important role, we will explain the
relationship here to provide context.
A pair of structures $(B,A)$ with $B\to A$ is said to be a \emph{gap pair} if $A\not\to B$, and every structure $C$
satisfying $B\to C$ and $C\to A$ is homomorphically equivalent to either $B$ or $A$~\cite{NesetrilTardif2000}. 
The same concept applies to structures with distinguished elements.
It is easy to see that any frontier for a structure $A$ must contain (modulo homomorphic equivalence) all structures $B$ such that $(B,A)$ is a gap pair.

\begin{example} Let $\mathcal{S}=\{R,P,Q\}$.
The structure $(A,a_1)$ consisting of facts $P(a_1)$ and $Q(a_1)$ (with distinguished element $a_1$) has a frontier of size 2 (w.r.t.~the class of all finite structures), namely $F=\{(B, a_1), (C, a_1)\}$ where $B$ consists of the facts $P(a_1),P(b),Q(b)$ and $C$ consists of the facts $Q(a_1),P(b),Q(b)$, respectively. Note that $((B, a_1), (A, a_1))$
and $((C, a_1), (A, a_1))$ are gap pairs. It can be shown that the structure $(A,a_1)$ has no frontier of size 1 (as such a frontier would have to consist of a structure that contains both facts $P(a_1)$ and $Q(a_1)$). 

For another example, consider the structure $(A',a_1)$ consisting of facts $P(a_1)$ and $R(b,b)$. It is the right hand side of a gap pair (the left hand side being the structure $(B', a_1)$ consisting of the facts $R(b,b)$ and $P(b')$), but $(A',a_1)$ has no frontier as follows from Theorem~\ref{thm:cacyclic} below. 
\end{example}

Frontiers are also closely related to (generalized) homomorphism dualities~\cite{FoniokNT08}, and we will be making use of results about homomorphism 
dualities. 
We say that a structure $(A,\textbf{a})$ \emph{has a finite duality} w.r.t.~
a class $\mathcal{C}$
if there is a finite set of structures $D$ such that for all
$(C,\textbf{c})\in\mathcal{C}$, 
$(A,\textbf{a})\to (C,\textbf{c})$ iff 
for all $(B,\textbf{b})\in D$, $(C,\textbf{c})\not\to (B,\textbf{b})$.  
The set $D$ may contain structures that are not in $\mathcal{C}$.
If ${\mathcal C}$ is the set of all structures (over the same schema as $(A,\textbf{a})$), we simply say that $(A,\textbf{a})$ \emph{has a finite duality}.\footnote{We note here that in the literature on Constraint Satisfaction, it is usual to consider the 'other side' of the duality, i.e,  a structure $A$ is said to have finite duality if there exists a finite set of structures $F$ such that for every structure $C$, $C\to A$ iff for all $B\in F$, $B\not\to C$.}

\begin{example} 
In the realm of digraphs, viewed as relational structures without distinguished elements with a single binary relation, every directed path $A$ of, say, $k>1$ nodes has finite duality (w.r.t.~the class of digraphs). Indeed, it is not difficult to verify that for every digraph $C$, $A\to C$ iff $C\not\to D$ where $D$ is the digraph with nodes $\{1,\dots k-1\}$ and edges $\{(i,j) \mid i<j\}$. (This example is known as the Gallai-Hasse-Roy-Vitaver Theorem. Amusingly, this result was obtained and published independently by all these four researchers, each in a different language, in the 1960s).
\end{example}

The next lemma is a minor variation of a result from~\cite{NesetrilTardif2000}.

\begin{lemma} \label{lem:frontiers-and-dualities}
Let $\mathcal{C}$ be any class of structures.
\begin{enumerate}
\item 
If a structure
$(A,\textbf{a})$ has a finite duality w.r.t. $\mathcal{C}$ then $(A,\textbf{a})$ has a  frontier w.r.t. $\mathcal{C}$.
\item 
If a structure
$(A,\textbf{a})\in\mathcal{C}$ has a frontier w.r.t. $\mathcal{C}$ and $\mathcal{C}$ is closed under direct products,  then  $(A,\textbf{a})$ has a finite duality w.r.t. $\mathcal{C}$.
\end{enumerate}
\end{lemma}

\begin{proof} 
  1.
Let $D$ be a finite set of structures that forms a duality for $(A,\textbf{a})$. Then $\{(A,\textbf{a})\times (B, \textbf{b}) \mid (B,\textbf{b})\in D\}$ is a frontier for $(A,\textbf{a})$.
This follows immediately from the fact that,
for all $(C,\textbf{c})$ with $(C,\textbf{c})\to (A,\textbf{a})$, we have that $(C,\textbf{c})\to (A,\textbf{a})\times (B, \textbf{b})$ if and only if $(C,\textbf{c})\to  (B,\textbf{b})$.

  2.
  We use a construction from~\cite{NesetrilTardif2000} involving
  an exponentiation operation on structures. Let $B$ and $C$ be structures (without distinguished elements) over the same schema. Then we denote by
  $B^C$ the structure where
  \begin{itemize}
      \item the domain of $B^C$ is the set of all functions from the domain of $C$ to the domain of $B$
      \item a fact $R(f_1, \ldots, f_n)$ belongs to $B^C$ if
      for every fact of the form $R(a_1, \ldots, a_n)\in C$, the fact
      $R(f_1(a_1), \ldots, f_n(a_n))$ belongs to $B$.
  \end{itemize}
  This construction is characterized by the property that, for all structures $D$,
  $D\to B^C$ if and only if $D\times C \to B$ \cite{HellNesetril2004}.
  
  Let $F$ be a frontier for $(A,\textbf{a})$. Let 
  \[D = \{(B^A,\textbf{h})\mid \text{$(B,\textbf{b})\in F$ and $\textbf{h}$ is a tuple of functions such that $h_i(a_i)=b_i$}\}\]
  We claim that $D$ forms a finite duality for $(A,\textbf{a})$. 
  Consider any structure $(C,\textbf{c})\in\mathcal{C}$. We must show that $(C,\textbf{c})$ homomorphically maps to a structure in $D$ iff $(A,\textbf{a})\not\to (C,\textbf{c})$.
  
  Suppose that there is a homomorphism $h: (C,\textbf{c})\to (B^A,\textbf{h})$ for some $(B^A,\textbf{h})\in D$.
  Then $\tilde{h}:(A,\textbf{a})\times (C, \textbf{c})\to (B,\textbf{b})$ where 
  $\tilde{h}(\langle a,c\rangle ) = h(c)(a)$. 
  Note that, indeed, $\tilde{h}(\langle a_i,c_i\rangle) = b_i$.
  Since $(B,\textbf{b})\in F$ and $F$ is a frontier, it follows that $(A,\textbf{a}) \not\to (A,\textbf{a})\times (C, \textbf{c})$ and 
  therefore, by the properties of direct products, 
  $(A,\textbf{a})\not\to (C,\textbf{c})$.
   
  Conversely, suppose $(A,\textbf{a})\not\to (C,\textbf{c})$. Then
  $(A,\textbf{a})\not\to(A,\textbf{a})\times (C, \textbf{c})$.
  Hence, there is a homomorphism 
  $h:(A,\textbf{a})\times (C, \textbf{c})\to (B,\textbf{b})$ for some $(B,\textbf{b})\in F$. 
  It follows that $(C, \textbf{c})\to (B^A,\textbf{h})$ where $\textbf{h}=h_1,\ldots, h_k$ with $h_i$ the function given by $h_i(x)=h(x,a_i)$.
  Note that $h_i(a_i)=b_i$ and hence $(B^A,\textbf{h})\in D$.  
\end{proof}

Note that the construction of the frontier from the duality is polynomial, while the construction of the duality from the frontier involves an exponential blowup. The following example shows that this is unavoidable.

\begin{example} The path
$\circ \xrightarrow{R}\circ \xrightarrow{R_1}\circ \xrightarrow{R}\circ
 \xrightarrow{R_2}\circ\cdots \circ  \xrightarrow{R_n}\circ
  \xrightarrow{R}\circ$, viewed as a structure without any distinguished elements, has a frontier (w.r.t.~the class of all finite structures) of size polynomial in $n$, as will follow from
  Theorem~\ref{thm:cacyclic-polyfrontier} below. It is known, however, that any finite duality for this structure 
  must involve a structure whose size is exponential in $n$, and 
  the example can be modified to use a fixed schema (cf.~\cite{NesetrilTardif2005}).
\end{example}


\subsection{Frontiers for classes with bounded expansion}
The notion of a \emph{class of graphs with bounded expansion} was introduced in~\cite{nesetril2018grad}. Intuitively, a class of graphs has bounded expansion 
if all of its shallow minors are sparse. 
We will not give a precise definition here, but  important examples include graphs of bounded degree, graphs of bounded treewidth, planar graphs, and any class of graphs excluding a 
minor. The same concept of bounded expansion can be applied also to arbitrary structures: 
a class of structures $\mathcal{C}$  is said to have bounded expansion if the class of \emph{Gaifman graphs} of structures in $\mathcal{C}$ has bounded expansion. We refer to~\cite{nesetril2012sparsity} for more details. Classes of structures of bounded expansion are in many ways computationally well-behaved (cf.~for example~\cite{kazana2020firsts}).

Ne\v{s}et\v{r}il and Ossona de Mendez \cite{NesetrilOssona2008, nesetril2012sparsity}
show that if $\mathcal{C}$ is any class of structures with bounded expansion, then every structure has a finite duality w.r.t. $\mathcal{C}$. It follows by Lemma~\ref{lem:frontiers-and-dualities} that also every structure has a fontier w.r.t.~$\mathcal{C}$.
Ne\v{s}et\v{r}il and Ossona de Mendez~\cite{NesetrilOssona2008, nesetril2012sparsity} only consider connected structures without distinguished elements,%
\footnote{Note that the various notions of connectedness, such as 
based on the incidence graph, fact graph, or Gaifman graph, all coincide for 
relational structures without distinguished elements.}
but their result extends in a straightforward way to 
the general case of structures with distinguished elements. Furthermore, 
it yields an effective procedure for constructing frontiers, although
non-elementary (i.e., not bounded by a fixed tower of exponentials).

\begin{theorem}[from \cite{NesetrilOssona2008, nesetril2012sparsity}]
\label{thm:frontiers-bounded-expansion}
Let $\mathcal{C}$ be any class of structures that has bounded expansion.
Then every structure $(A,\textbf{a})$ has a frontier w.r.t. $\mathcal{C}$, which can be effectively constructed. 
\end{theorem}

\begin{proof} 
Ne\v{s}et\v{r}il and Ossona de Mendez~\cite{NesetrilOssona2008} stated their result for the case where $A$ is a connected structure without distinguished elements. They show that, every such structure $A$ has a finite duality w.r.t.~$\mathcal{C}$, and hence, by Lemma~\ref{lem:frontiers-and-dualities}, also a fontier w.r.t.~$\mathcal{C}$. The result extends to disconnected structures through standard arguments: let $A$ be any structure (without distinguished elements). We may assume without loss of generality that $A$ is a core (because every structure is homomorphically equivalent to its core, and hence every frontier for the latter is a frontier for the former).
Let $A_1, \ldots, A_n$ be the connected components of $A$.  
Since $A$ is a core, $A_1, \ldots, A_n$ are pairwise homomorphically incomparable.
Take all structures of the form $A_1\uplus \cdots \uplus A_{i-1}\uplus B \uplus A_{i+1} \uplus \cdots \uplus A_n$, for $B$ a structure belonging to the frontier of $A_i$ (for $i=1\ldots n$). It is 
straightforward to show that this yields a frontier for $A$. 

Next, we show how to extend this to structures with distinguished elements.
For any structure $(A,\textbf{a})$  with $\textbf{a}=a_1, \ldots, a_n$, 
let $A^{\textbf{a}}$ be the structure (without distinguished elements) over expanded schema with additional  unary predicates $P_1, \ldots, P_n$, 
where each $P_i$ denotes $\{a_i\}$. Since $A$ and $A^{\textbf{a}}$ have the same
Gaifman graph, and since $\mathcal{C}$ has bounded expansion, the same holds for
$\{A^{\textbf{a}}\mid (A,\textbf{a})\}$. Therefore, it suffices to show that 
whenever $A^{\textbf{a}}$ has a frontier w.r.t.~$\{C^{\textbf{c}}\mid (C,\textbf{c})\in \mathcal{C}\}$, then $(A,\textbf{a})$ has a frontier w.r.t.~$\mathcal{C}$.

Let $F=\{B_1, \ldots, B_m\}$ be a frontier for $A^{\textbf{a}}$ w.r.t.~$\{C^{\textbf{c}}\mid (C,\textbf{c})\in \mathcal{C}\}$.
Now consider all ways of taking a structure $B\in F$ and choosing one element per unary predicate $P_i$. 
In this way we obtain a set of structures $F'$  with distinguished elements, that we claim is a frontier for $(A,\textbf{a})$ w.r.t.~$\mathcal{C}$.
Note that any homomorphism $h:(A,\textbf{a})\to (B,\textbf{b})$ for $(B,\textbf{b})\in F'$ is a homomorphism from $A^{\textbf{a}}$ to  $B^{\textbf{b}}$, therefore since $F$ is a frontier for $A^{\textbf{a}}$, there is no such homomorphism $h$.
It is also clear that each $(B,\textbf{b})\in F'$ maps homomorphically to  $(A,\textbf{a})$. 
Finally, consider any $(C,\textbf{c})\in\mathcal{C}$ that maps to $(A, \textbf{a})$ but not vice versa. Then 
$C^{\textbf{c}}\to A^{\textbf{a}}$ and $A^{\textbf{a}}\not\to C^{\textbf{c}}$. Hence, there is a homomorphism $h:C^{\textbf{c}}\to B$ for some $B\in F$. It follows that
$h:(C,\textbf{c})\to (B,\textbf{b})\in F'$ where each $b_i=h(c_i)$.
\end{proof}


\subsection{Polynomial frontiers for c-acyclic structures}\label{sec:frontier-cacyclic}

Alexe et al.~\cite{AlexeCKT2011}, building on Foniok et al.~\cite{FoniokNT08}, show that a structure has a finite duality if and only if its core is c-acyclic. By Lemma~\ref{lem:frontiers-and-dualities}, this implies that a structure has a frontier if and only if its core is c-acyclic. 

\begin{theorem}[from~\cite{FoniokNT08,AlexeCKT2011}]
\label{thm:cacyclic}
  For all structures $(A,\textbf{a})$, the following are equivalent:
  \begin{enumerate}
      \item $(A,\textbf{a})$ has a frontier w.r.t.~the class of all structures,
      \item $(A,\textbf{a})$ is homomorphically equivalent to a c-acyclic structure,
      \item The core of $(A,\textbf{a})$ is c-acyclic
  \end{enumerate}
\end{theorem}

One of our main results is  a new proof of the right-to-left direction, which, unlike the original, provides a polynomial-time construction of a frontier from a c-acyclic structure:

\begin{theorem}\label{thm:cacyclic-polyfrontier}
Fix a schema $\mathcal{S}$ and $k\geq 0$.
Given a c-acyclic structure over $\mathcal{S}$ with $k$ distinguished elements, we can construct in polynomial time a frontier w.r.t. the class of all  structures over $\mathcal{S}$ that have
$k$ distinguished elements.
\end{theorem}

Note that the size of the smallest frontier is in general exponential in $k$. Indeed, consider the single-element structure $(A,\textbf{a})$ where $A$ consists of the single fact $P(a)$ and 
$\textbf{a}=a,\ldots,a$ has length $k$. It is not hard to show that every frontier of this
(c-acyclic) structure must contain, up to homomorphic equivalence, all structures of the 
form $(B,\textbf{b})$ where $B$ consists of two facts, $P(a_1)$ and $P(a_2)$, 
and $\textbf{b}\in \{a_1, a_2\}^k$ is a sequence in which both $a_1$ and $a_2$ 
occur. There are exponentially many pairwise homomorphically incomparable such structures.

The proof of Theorem~\ref{thm:cacyclic-polyfrontier} is based on a 
construction that improves over a similar but exponential construction of gap pairs for acyclic structures given in~\cite[Def.~3.9]{NesetrilTardif2000}.
Our results also shed new light on a question posed in the same paper: after presenting a double-exponential construction of duals (for connected structures without distinguished elements), involving first constructing an exponential-sized gap pair, 
the authors ask: ``\emph{It would be  interesting to know to what extent the characterisation of duals can be simplified, and whether the indirect approach via density is optimal}.'' This question appeared to have been answered in~\cite{NesetrilTardif2005}, where a direct method was established for constructing single-exponential size duals.
Theorem~\ref{thm:cacyclic-polyfrontier} together with Lemma~\ref{lem:frontiers-and-dualities}, however, gives
another answer: single-exponential duals can be constructed by going through frontiers (i.e., ``via density'') as well.

Recall the definition of fg-connectedness from the preliminaries.
We first prove a restricted version of Theorem~\ref{thm:cacyclic-polyfrontier} for the special case of core, fg-connected, c-acyclic structures with the Unique Names Property. We subsequently lift these extra assumptions.  A
structure $(A, \textbf{a})$ with $\textbf{a}=a_1, \ldots a_k$ has the \emph{Unique Names Property (UNP)} if $a_i\neq a_j$ for all $i< j$ (cf.~\cite{baader2003basic}).

\begin{proposition}\label{prop:fgconnected-cacyclic-unp}
Given a core, fg-connected, c-acyclic structure with UNP, we can construct in polynomial time (for fixed schema $\mathcal{S}$ and number of distinguished elements $k$) a frontier w.r.t. the class of all finite structures. Furthermore, the frontier consists of a single structure, which has the UNP.
\footnote{
An earlier conference version of this paper had a bug in the proof of this proposition, as was pointed out to us by 
Raoul Koudijs
(p.c.).}
\end{proposition}

\begin{proof} 
Let a core fg-connected c-acyclic structure $(A, \textbf{a})$ with UNP be given. 
To reduce notational complexity in the remainder of this proof, we will simply write $A$ instead of 
$(A, \textbf{a})$, even when referring to the structure including the distinguished elements.

Note that each fg-connected structure either (i) consists of a single fact containing only distinguished elements, or (ii) consists of a number of facts that all contain at least one non-distinguished element.
Therefore, we can distinguish two cases:

Case 1. $A$ consists of a single fact $f$ without non-distinguished elements. Let $(B, \textbf{a})$ be the structure whose domain is 
$\{a_1, \ldots, a_k, b\}$, where $\textbf{a}=a_1,\ldots,a_k$ and $b$ is a fresh value distinct from $a_1, \ldots, a_k$; and which contains all facts over this domain except $f$. It is easy to see that $(B,\textbf{a})$ is a homomorphism dual for $(A,\textbf{a})$.
Indeed, consider any structure $(C,\textbf{c})$, and let $f'$ be a copy of the fact
$f$ in which each element $a_i$ is replaced by the corresponding element $c_i$.
If $h:(A,\textbf{a})\to (C,\textbf{c})$ then
$C$ contains $f'$, therefore, $(C,\textbf{c})\not\to (B,\textbf{a})$; 
if, on the other hand, $(A,\textbf{a})\not\to (C,\textbf{c})$, then $C$ omits $f'$,
and hence, $(C,\textbf{c})\to (B,\textbf{a})$.
It follows that, the direct product $(A,\textbf{a})\times (B,\textbf{a})$
constitutes a singleton frontier for $(A, \textbf{a})$. Note that this construction is polynomial because we assume that the schema $\mathcal{S}$ and $k$ are both fixed.

Case 2. $A$ consists of one or more facts that each contain a non-distinguished element. In this case,
we construct a singleton frontier $F=\{(B,\textbf{b})\}$ where 
\begin{itemize}
    \item 
the domain of $B$ consists of  
\begin{enumerate} \item all pairs $(a,f)$ where $a$ is a non-distinguished element of $A$ and $f$ is a fact of $A$ in which $a$ occurs, and
\item All pairs $(a, \textsf{id})$ and $(a,\textsf{nd})$, where $a$ is a distinguished element of $A$ 
\end{enumerate}
    \item a fact $R((a_1,f_1),\ldots, (a_n,f_n))$ holds in $B$ if and only if 
     $R(a_1, \ldots, a_n)$ holds in $A$ and 
     at least one $f_i$ is either a fact that is different from the fact $R(a_1, \ldots, a_n)$ itself, or is $\textsf{nd}$
\item The distinguished elements $\textbf{b}$ are $b_1=(a_1, \textsf{id}), \ldots, b_n=(a_n, \textsf{id})$, for $\textbf{a}=a_1, \ldots, a_n$.
\end{itemize}
Note that, in the above construction, $\textsf{id}$ and $\textsf{nd}$ are symbols (not functions), used to simplify notation by ensuring that every element of $B$ can be written as a pair. The symbols $\textsf{id}$ and $\textsf{nd}$, intuitively, stand for ``identity'' and ''non-distinguished copy'', 

We claim that $F=\{(B,\textbf{b})\}$ is a frontier for $A$.

It is clear that the natural projection $h:(B,\textbf{b})\to (A,\textbf{a})$ is a homomorphism.

We claim that there is no homomorphism $h':(A,\textbf{a})\to (B,\textbf{b})$. Assume, for the sake of a contradiction, that there was such a homomorphism.
By Lemma~\ref{lem:core-injective}, we may assume that 
the composition of $h$ and $h'$ is the identity function on $A$. 
In particular, this means that $h'$ maps each distinguished element $a$ to $(a, \textsf{id})$
and for each non-distinguished element $a$ of $A$, $h'(a) = (a,f)$ for some fact $f$.
For a non-distinguished element $a$, let us denote by $f_a$ the unique fact 
$f$ for which $h'(a)=(a,f_a)$.

We will consider ``walks'' in $A$ of the form 
 $$ a_1 \xrightarrow{f_{a_1}} a_2 \xrightarrow{f_{a_2}} \ldots a_n$$
with $n\geq 1$, where 
\begin{enumerate}
  \item $a_1, \ldots, a_n$ are non-distinguished elements,
\item $f_{a_i} \neq f_{a_{i+1}}$, and
\item $a_i$ and $a_{i+1}$ co-occur in fact $f_{a_i}$,
\end{enumerate}
Since $A$ is c-acyclic, the length of any such sequence is bounded by the diameter of $A$ (otherwise some fact would have to be traversed twice in succession, which would violate condition 2). 
Furthermore,  trivially, such a walk of length $n=1$ exists: just choose as $a_1$ an arbitrary non-distinguished element of $A$. Furthermore, we claim that any such finite sequence can be extended to a longer one: let the fact $f_{a_n}$ be 
of the form $R(b_1, \ldots, b_m)$ (where $a_n = b_i$ for some $i\leq m$).
Since $h$ is a homomorphism, it must map 
$f_{a_n}$
to some fact
$R((b_1,f_{b_1}), \ldots, (b_m, f_{b_m}))$ of $B$, where 
some $f_{b_j}$ is a fact that is different from $f_{a_n}$. We can choose $b_j$ as our element $a_{n+1}$.
Thus, we reach our desired contradiction.

Finally, consider any $C$ with $h:C\to A$ and $A\not\to C$.
We construct a function $h':C\to B$ as follows: consider any element $c$ of $C$, and
let $h(c)=a$. If $c$ is a distinguished element (in which case $a$ is, too), we set $h'(c)=(a,\textsf{id})$. If $c$ is not a distinguished element but $a$ is, we set
$h'(c)=(a,\textsf{nd})$. Otherwise, 
we proceed as follows: since $A$ is c-acyclic and fg-connected, for each non-distinguished element $a'$ of $A$ (other than $a$ itself) there is a unique minimal path in the incidence graph, containing only non-distinguished elements, from $a'$ to $a$. We can represent this 
path by a sequence of the form \[a'=a_0 \xrightarrow{(f_0,i_0,j_0)} a_1 \xrightarrow{(f_1,i_1,j_1)} a_2 \cdots \xrightarrow{(f_{n-1},i_{n-1},j_{n-1})} a_n = a\]
where each $f_\ell$ is a fact of $A$ in which $a_\ell$ occurs in the $i_\ell$-th position
and $a_{\ell+1}$ occurs in the $j_\ell$-th position.
We can partition the non-distinguished elements $a'$ of $A$  (other than $a$ itself) according to the last fact on this path, that is, $f_{n-1}$.
Furthermore, it follows from fg-connectedness that each fact of $A$ contains a non-distinguished element. It is easy to see that if a fact contains multiple non-distinguished elements (other than $a$) then they must all belong to the same part of the partition as defined above. Therefore, the above partition on non-distinguished elements
naturally extends to a partition on the facts of $A$. Note that if $A$ contains any facts 
in which $a$ is the only non-distinguished element, we will refer to these facts as ``local facts'' and they will be handled separately.
In this way, we have essentially decomposed $A$ into a union
$A_{\text{local}}\cup\bigcup_i A_i$, where $A_{\text{local}}$ contains all local facts and each ``component'' $A_i$ is a substructure of $A$ consisting of non-local facts, in such a way that
(i) different substructures $A_i$ do not share any facts with each other,
(ii) different substructures do not share any elements with each other, except for 
$a$ and distinguished elements (from $\textbf{a}$), (iii) each $A_i$ contains precisely one fact involving $a$. 

Since we know that $(A, a)\not\to (C, c)$, 
it follows that either some local fact $f$ of $A$ does not map to $C$ (when sending $\textbf{a}, a$ to $\textbf{c}, c$), or some ``component'' $A_i$ of $A$ does not map to $C$ through any homomorphism sending $\textbf{a}, a$ to $\textbf{c}, c$. 
In the first case, we choose such local fact $f$ and set $h'(a) = (a,f)$. In the second case, we choose such a component
(if there are multiple, we choose one of minimal size)
and let $f$ be the unique fact in that component containing $a$ (that is, $f$ is the 
fact $f_{n-1}$ that by construction connected the non-distinguished elements of the component in question to $a$). We set $h'(c) = (a,f)$. 
Intuitively, when $h'(c) = (h(c), f)$, then $f$ is a fact of $A$ involving $h(c)$ that ``points in a direction where homomorphism from $(A, h(c))$ back to $(C, c)$ fails''. 

We claim that $h'$ is a homomorphism from $C$ to $B$: let $R(c_1, \ldots, c_n)$ be a fact of $C$. Then $R(h(c_1), \ldots, h(c_n))$ holds in $A$. 
Let $h'(c_i)=(h(c_i), f_i)$ as constructed above (where, we recall, $f_i=\textsf{id}$ if $c_i$ is a distinguished element, and $f_i=\textsf{nd}$ if $c_i$ is not a distinguished element but $h(c_i)$ is).
Also recall that at least one $c_i$ is a non-distinguished element. 
To show that 
$R(h'(c_1), \ldots, h'(c_n))$ holds in $B$, it suffices to show that
some $f_i$ is different from the fact $R(h(c_1), \ldots, h(c_n))$ itself, or is equal to $\textsf{nd}$.
If some non-distinguished $c_i$ is mapped by $h$ to a distinguished element, then 
$h'(c_i)=(c_i,\textsf{nd})$, and we are done. This leaves us with the case where
some $c_i$ is a non-distinguished element, and for all non-distinguished $c_i$,
$h'(c_i)$ is of the form $(h(c_i),f_i)$ for a fact $f_i$.
If one of these $f_i$ is a local fact, then it follows immediately from the construction that
$f_i \neq R(h(c_1), \ldots, h(c_n))$.
Otherwise, let $n_i$ be the size of the smallest ``component'' (as defined above) of $(A,h(c_i))$ that does not homomorphically map to $(C,c_i)$, and choose an element $c_i$ with minimal $n_i$. Then, clearly, $f_i$ must be different from the fact $R(h(c_1), \ldots, h(c_n))$ itself.
\end{proof}

Next, we remove the assumptions of fg-connectedness and being a core.

\begin{proposition}\label{prop:unp-disconnected}
Given a c-acyclic structure with UNP, we can construct in polynomial time a frontier w.r.t. the class of all finite structures. Furthermore, the frontier consists of structures that have the UNP.
\end{proposition}

\begin{proof}
By Proposition~\ref{prop:cacyclic-core}, we may assume that $(A, \textbf{a})$ is a core. 
Note that the c-acyclicity and UNP properties are preserved under the passage from a structure to its core.

Let $(A, \textbf{a})$ be a structure with distinguished elements that is UNP and that is a fg-disjoint union of homomorphically incomparable fg-connected structures
$(A_1,\textbf{a}), \ldots, (A_n, \textbf{a})$. 
By Proposition~\ref{prop:fgconnected-cacyclic-unp},  $(A_1,\textbf{a}), \ldots, (A_n, \textbf{a})$ have, respectively, frontiers 
 $F_1, \ldots, F_n$, 
each consisting of a single structure with the UNP.
We may assume without loss of generality that each $F_i$ consists of a structure that have the same distinguished elements $\textbf{a}$ (we know that the structures in question have the UNP, and therefore, modulo isomorphism, we can assume that the distinguished elements are precisely $\textbf{a}$).
Let $F_i=\{(B_i,\textbf{a})\}$.

We claim that 
$F = \{ \big(\biguplus_{j\neq i} (A_j,\textbf{a})\big) \uplus (B_i,\textbf{a}) \mid 1\leq i\leq n\}$
is a frontier for $(A,\textbf{a})$ w.r.t.~$\mathcal{C}$.

Clearly, each structure in $F$ maps homomorphically to $A$.

Suppose, for the sake of contradiction, that there is a homomorphism $h: (A,\textbf{a}) \to \big(\biguplus_{j\neq i} (A_j,\textbf{a})\big) \uplus (B_i,\textbf{a})$ for some $i$. Observe that $h$ must send each distinguished element to itself, and it must send each non-distinguished element to a non-distinguished element (otherwise, the composition of $h$ with the backward homomorphism would be a non-injective endomorphism on  $(A,\textbf{a})$ which would contradict the fact that $(A,\textbf{a})$ is a core).
Since $(A_i,\textbf{a})$ is fg-connected (and because $h$ cannot send non-distinguished elements to distinguished elements), its $h$-image must be contained either in some $(A_j,\textbf{a})$ ($j\neq i$) or in $B$. 
The former cannot happen because $A_i$ and $A_j$ are homomorphically incomparable.
The latter cannot happen either, because $B$ belongs to a frontier of $A_i$.

Finally, let $(C,\textbf{c})\in\mathcal{C}$ be any structure such that there is a homomorphism $h:(C,\textbf{c})\to (A,\textbf{a})$ but $(A,\textbf{a})\not\to (C,\textbf{c})$.
Let $(A_i,\textbf{a})$ be a fg-connected component of $(A,\textbf{a})$ such that $(A_i,\textbf{a})\not\to (C,\textbf{c})$.
Since $(A_i,\textbf{a})$ is fg-connected, we can partition our structure $(C,\textbf{c})$ as $(C_1,\textbf{c})\uplus (C_2,\textbf{c})$ where the $h$-image of $C_1$ is contained in $(A_i,\textbf{a})$ while the $h$-image of $C_2$ is disjoint from $A_i$ except possibly for the 
distinguished elements.
We know that $(A_i,\textbf{a})\not\to (C_1,\textbf{c})$ and therefore $(C_1,\textbf{c})\to (B_i,\textbf{a})$. 
Furthermore, we have that $(C_2,\textbf{c})\to \biguplus_{j\neq i} (A_j,\textbf{a})$. Therefore,
$(C,\textbf{c})\to \big(\biguplus_{j\neq i} (A_j,\textbf{a})\big) \uplus (B_i,\textbf{a})$.
\end{proof}

Finally, we can prove Theorem~\ref{thm:cacyclic-polyfrontier} itself.

\begin{proof}[Proof of Theorem~\ref{thm:cacyclic-polyfrontier}]
Let $(A,\textbf{a})$ be c-acyclic. 
If it has the UNP, we are done. Consider the other case, where the sequence $\textbf{a}$ contains repetitions. 
Let $\textbf{a}' = a'_1, \ldots, a'_n$ consists of the same elements without repetition (in some order). 
We construct a frontier for it as follows:
\begin{enumerate}
  \item Consider the structure $(A, \textbf{a}')$, which, by construction, has the UNP. Let $F$ be a frontier for $(A, \textbf{a}')$ (again consisting of structures with the UNP), using Proposition~\ref{prop:unp-disconnected}. Note that, through isomorphism, we may assume that each structure in $F$ has the same 
  distinguished elements $\textbf{a}'$. For each  $(B,\textbf{a}')\in F$, we take the structure $(B,\textbf{a})$. 
  \item Let  $k$ be the length of the tuple $\textbf{a}$.
  For each function $f:\{1, \ldots, k\}\to \{1, \ldots, k\}$,  whose range has size strictly greater than $n$, consider structure $(C, \textbf{c}^f)$ where $C$ contains all facts over the domain $\{1, \ldots, k\}$, and $c^f_i = f(i)$. We take its direct product with $(A,\textbf{a})$.
\end{enumerate}
It is easy to see that the set of all structures constructed above, constitutes a frontier for $(A, \textbf{a})$. Indeed, suppose a structure maps to $(A,\textbf{a})$ but not vice versa. 
If the tuple of distinguished elements of the structure in question has the same identity type as the tuple $\textbf{a}$
(i.e., the same equalities hold between values at different indices in the tuple)
then it is easy to see that the structure in question must map to some structure $(B, \textbf{a})$ as constructed under item 1 above. Otherwise, if
the tuple of distinguished elements of the structure in question does \emph{not} have the same identity type, then it is easy to see that the structure in question must map to $(C^f, \textbf{c}^f)\times (A, \textbf{a})$, as constructed under item 2 above, where $f$ reflects the identity type of the distinguished elements of the structure in question.
\end{proof}

As a corollary of Theorem~\ref{thm:cacyclic-polyfrontier}, we obtain the following
interesting by-product:

\begin{theorem} \label{thm:testing-frontier}
Fix a schema $\mathcal{S}$ and $k\geq 0$.
The following problem is solvable in NP: given a finite set of structures $F$ and a structure $A$ (all with $k$ distinguished elements), is $F$ a frontier for $A$ w.r.t.~the class of all structures? 
If $\mathcal{S}$ contains a binary relation and $k=3$, then it is NP-complete.
\end{theorem}

\begin{proof}
For the upper bound, we use the fact that, 
if  $A$ is homomorphically equivalent to a c-acyclic structure $A'$,
then the core of $A$ is c-acyclic (cf.~Theorem~\ref{thm:cacyclic}).
The problem can therefore be solved in non-deterministic polynomial time as follows:

First we guess a substructure $A'$ and we verify that $A'$
is c-acyclic and homomorphically equivalent to $A$. Note that the existence of such $A'$ is a necessary precondition for $F$ to be a frontier of $A$. Furthermore, c-acyclicity can be checked in polynomial time using any PTIME algorithm for graph acyclicity (recall that a structure is c-acyclic if and only if its incidence graph is acyclic after removing all nodes corresponding to distinguished elements).

Next, we apply Theorem~\ref{thm:cacyclic-polyfrontier} to construct a frontier $F'$ for $A'$ (and hence for $A$).
Finally, we verify that each $B\in F$ homomorphically maps to some $B'\in F'$ and, vice versa, every $B'\in F'$ homomorphically maps to some $B\in F$. It is not hard to see that this non-deterministic algorithm has an accepting run if and only if $F$ is a frontier for $A$.

For the lower bound, we reduce from graph 3-colorability. Let $A$ be the structure,
over a 3-element domain, that consists of the facts $R(a,b)$ for all pairs $a,b$ with
$a\neq b$. In addition, each of the three elements is named by a constant symbol.
Since $A$ is c-acyclic, by Theorem~\ref{thm:cacyclic}, it has a frontier $F$.
Now, given any graph $G$ (viewed as a relational structure with binary relation $R$ and without constant symbols), we have that $G$ is 3-colorable if and only 
if $F$ is a frontier for the disjoint union of $A$ with $G$. To see that this is the case, note that if $G$ is 3-colorable, then the disjoint union of $A$ with $G$ is homomorphically equivalent to $A$ itself, whereas if $G$ is not 3-colorable, then the disjoint union of $A$ with $G$ is strictly greater than $A$ in the homomorphism order.
\end{proof}

\subsection{A polynomially frontier-closed class of structures}\label{sec:frontier-closed}
We call a class $\mathcal{C}$ of structures  \emph{frontier-closed} if every structure $(A,\textbf{a})\in\mathcal{C}$ has a frontier w.r.t.~$\mathcal{C}$, consisting of structures belonging to $\mathcal{C}$.
If, moreover, the frontier in question can be constructed from $(A,\textbf{a})$ in polynomial time, then we say that $\mathcal{C}$ is \emph{polynomially frontier-closed}. 

\begin{theorem}\label{thm:frontier-closed}
Fix a schema $\mathcal{S}$ and $k\geq 1$.
The class of c-connected acyclic structures with $k$ distinguished elements is polynomially frontier-closed.~\footnote{Note that for structures with one distinguished element,
c-connectedness is the same as connectedness.}
\end{theorem}


In fact, the construction presented below shows that the polynomial bound holds even when the schema is treated as part of the input of the problem (although $k$ does need to be fixed, as the size of the constructed frontier depends exponentially on $k$).

As will follow from results in Section~\ref{sec:characterizations} (cf.~Theorems~\ref{thm:unary-is-necessary}-\ref{thm:acyclicity-is-necessary} below) the theorem fails if we drop any of the three restrictions in the statement (i.e., c-connectedness, acyclicity, and $k\geq1$).

\subsection*{The special case with binary relations only and $k=1$}

The remainder of this section is dedicated to the proof of Theorem~\ref{thm:frontier-closed}.
To simplify the presentation of the proof, we will first assume that the schema consists of binary relations only, and that $k=1$. Afterwards, we will show how to lift these restrictions. 

Let $\mathcal{S}$ be a schema consisting of binary relation symbols, and fix a finite structure $(A,a_0)$ that is \textbf{c-connected} and \textbf{acyclic}. We will assume, in addition, that $(A,a_0)$ is a core, which we may do without loss of
generality, because the core of an acyclic structure can be computed in polynomial time (Proposition~\ref{prop:cacyclic-core}) and the properties of c-connectedness and acyclicity are preserved under passage from a structure to its core:

\begin{proposition}
The properties of c-connectedness and acyclicity are preserved when passing from a structure to its core.
\end{proposition}

\begin{proof}
That acyclicity is preserved follows immediately from the fact that the core is a substructure of the original structure. For c-connectedness, the argument is as follows:
let $(B,\textbf{b})$ be a c-connected structure, and let $(B',\textbf{b})$ be its core. 
It follows from the definition of a core that $(B,\textbf{b})\leftrightarrow (B',\textbf{b})$. 
By Proposition~\ref{prop:reach}, this implies that $(B,\textbf{b})\leftrightarrow (B',\textbf{b})^{\text{reach}}$, and hence, $(B',\textbf{b})\leftrightarrow (B',\textbf{b})^{\text{reach}}$.
It follows by the minimality property of cores that 
$(B,\textbf{b}) = (B',\textbf{b})^{\text{reach}}$, i.e.,  $(B',\textbf{b})$ is c-connected.
\end{proof}

We shall slightly abuse notation and, for every relation symbol $R$ and every $a,b\in A$ we shall say that $R^-(a,b)$ holds in (or is a fact of) $A$ if $R(b,a)$ is a fact of $A$.
We can think of $A$ as an (oriented) tree rooted at $a_0$, where every edge $b\to c$ has been oriented away from $a_0$ and is labelled $R$ or $R^-$ depending on whether $R(b,c)$ or $R(c,b)$ is a fact of $A$.

\begin{definition}[$A|a$] For any element $a$ of $A$,
we denote by $A|a$ the substructure of $A$ consisting of the oriented subtree rooted at $a$. 
\end{definition}




\begin{definition}[rank]
For any element $a$ of $A$,
$rank(a)$ is the depth of the oriented tree $A|a$. 
\end{definition}

The construction of frontiers that we will describe below,  involves a recursion on rank.
Note that $rank(a)=0$ when $a$ is maximally far from $a_0$ (in other words, the leafs
of the oriented tree have rank 0).
Furthermore, if there is an edge $a\to b$ (in the oriented tree) then 
$rank(b) < rank(a)$. 

The frontier construction below is split into two steps:
for each element $a$, we first construct a set of structures $\mathcal{F}_a$, and, 
subsequently, we construct a modified set of structures $\mathcal{F}^*_a$.

\begin{definition}[$\mathcal{F}_a$]
\label{def:F}
Fix any node $a$ of $A$. Let $a'_1, \ldots, a'_n$ be the children of $a$ (in the oriented tree), and let $S_1, \ldots, S_n$ be the corresponding edge labels.  We can depict $A|a$ 
as follows:
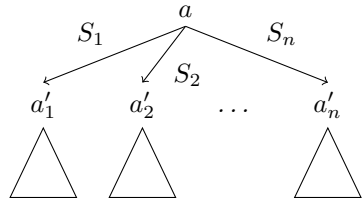
\begin{center}
\begin{tikzpicture}[level distance=30pt,sibling distance=12pt]
\Tree [.{$a$}
\edge[->] node[auto=right]{$S_1$}; 
  [.\text{$a'_1$} \edge[roof] node[auto=left]{}; {~~~~~~} ]
\edge[->] node[auto=left]{$S_2$}; 
  [.\text{$a'_2$} \edge[roof] node[auto=left]{}; {~~~~~~} ]
\edge[draw=none] ; 
  [.\text{\dots} ]
\edge[->] node[auto=left]{$S_n$}; 
  [.\text{$a'_n$} \edge[roof] node[auto=left]{}; {~~~~~~} ] ]
\end{tikzpicture}
\end{center}
%
If $rank(a)=0$ (that is, if $n=0$)  we define $\mathcal{F}_a=\emptyset$. Otherwise, we define $\mathcal{F}_a =\{ H^1, \ldots, H^n\}$ where
   $H^i$ is obtained from $A|a$ in the following way. Set $H^i$ to be the result of removing $A|a'_i$ from $A|a$. Then, we add to $H^i$ a fresh isomorphic copy $F$ of each structure in $\mathcal{F}_{a'_i}$, joining $a$ with the newly created copies of $a'_i$ with a $S_i$-edge.
See Figure~\ref{fig:illustration-F} for an illustration.
\end{definition}

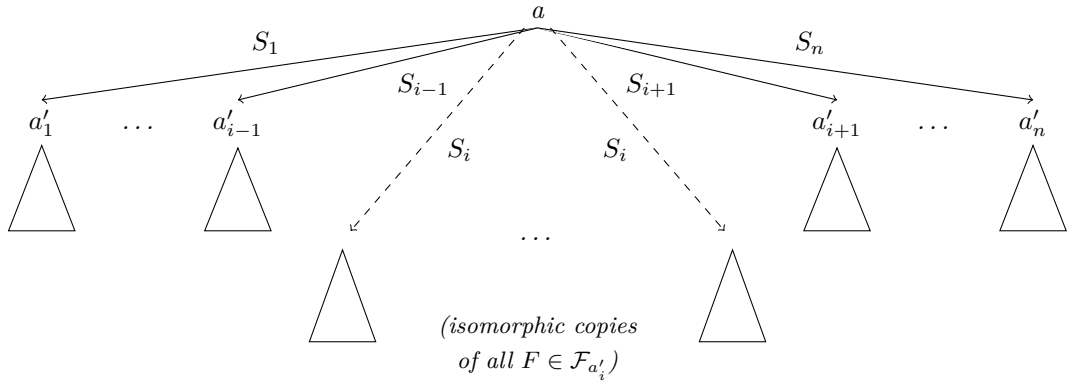
\begin{figure}
\begin{center}
\begin{tikzpicture}[level distance=35pt,sibling distance=14pt]
\Tree [.\node(root){a} ;
\edge[->] node[auto=right]{$S_1$}; 
  [.\text{$a'_1$} \edge[roof] node[auto=left]{}; {~~~~~~} ]
\edge[color=white] ; 
  [.\text{\dots} ]
\edge[->] node[auto=left]{$S_{i-1}$}; 
  [.\text{$a'_{i-1}$} \edge[roof] node[auto=left]{}; {~~~~~~} ]
\edge[draw=none] ; 
  [.\text{} \edge[draw=none] ; 
    [.\node(ff){} ; \edge[roof,fill=black] node[auto=left]{}; {~~~~~~} ] ]
\edge[draw=none] ; 
  [.\text{} \edge[draw=none] ; 
    [.\text{\dots} \edge[draw=none] node[auto=left]{}; {\begin{tabular}{@{}c@{}}(an isomorphic copy of \\ of each $F\in \mathcal{F}_{a'_i}$)\end{tabular}} ] ]
\edge[draw=none] ; 
  [.\text{} \edge[draw=none] ; 
    [.\node(fl){} ; \edge[roof,fill=black] node[auto=left]{}; {~~~~~~} ] ]
\edge[->] node[auto=right]{$S_{i+1}$}; 
  [.\text{$a'_{i+1}$} \edge[roof] node[auto=left]{}; {~~~~~~} ]
\edge[color=white] ; 
  [.\text{\dots} ]
\edge[->] node[auto=left]{$S_n$}; 
  [.\text{$a'_n$} \edge[roof] node[auto=left]{}; {~~~~~~} ] ]
\draw[->,dashed] (root) -- (ff) node[midway,auto=left] {$S_i$};
\draw[->,dashed] (root) -- (fl) node[midway,auto=right] {$S_i$};
\end{tikzpicture}

\end{center}
\caption{A depiction of the construction of $\mathcal{F}_a$ in Definition~\ref{def:F}.}
\label{fig:illustration-F}
\end{figure}
 
Observe that, for $F\in \mathcal{F}_a$, there is a natural mapping  $h:F\to A$
(indeed, following the induction of the construction, we can see that each element
of $F$ is either an element of $A$ or else was introduced as an isomorphic copy of an element of $A$).
 
 \begin{definition}[$\mathcal{F}^*_a$]
 \label{def:F*}
 We define $\mathcal{F}^*_a = \{ F^* \mid F\in\mathcal{F}_a\}$, where
 $F^*$ is defined as follows. Let $h:F\to A$ be the natural mapping. 
Then, $F^*$ is obtained from $F$ by adding, for each element $b$ of $F$ with $h(b)\neq a$,
 a fresh isomorphic copy of $A$ together with a connecting $S$-edge from $c$ to $b$, 
where $c$ is  (the newly created copy of) the parent of $h(b)$ and $S$ is the edge label of the edge from $c$ to $h(b)$ in $A$.
 \end{definition}
  
\begin{figure}
\begin{tabular}{ccc}
\begin{tikzpicture}[level distance=1cm,sibling distance=1cm, ->, baseline={(a0)}] 
\Tree
[.\node[draw,circle](a0){$a_0$}; \edge ;
[.{$a_1$} \edge ;
[.{$a_2$} \edge ;
[.{$a_3$} \edge ;
[.{$a_4$} 
]]]]]
\end{tikzpicture}
~~~~~~&~~~~~~
\begin{tikzpicture}[level distance=1cm,sibling distance=1cm, ->, baseline={(a0)}]
\Tree
[.\node[draw,circle](a0){$a_0$}; \edge ;
[.{$a_1$} \edge ;
[.{$a_2$} \edge ;
[.{$a_3$} 
]]]]
\end{tikzpicture}
~~~~~~&~~~~~~
\begin{tikzpicture}[level distance=1cm,sibling distance=1cm, ->, baseline={(a0)}]
\Tree
[.\node[draw,circle](a0){$a_0$}; \edge[thick] ;
[.\node(a1){$a_1$}; \edge[thick] ;
[.\node(a2){$a_2$}; \edge[thick] ;
[.\node(a3){$a_3$};
]]]]

\begin{scope}[xshift=1cm,yshift=.75cm,level distance=.8cm, sibling distance=0.5cm, font=\scriptsize]
\Tree
[.\node(a0prime){$a'_0$}; \edge ;
[.{$a'_1$} \edge ;
[.{$a'_2$} \edge ;
[.{$a'_3$} \edge ;
[.{$a'_4$} 
]]]]]
\end{scope}

\begin{scope}[xshift=-1cm,yshift=0cm,level distance=.8cm, sibling distance=0.5cm, font=\scriptsize]
\Tree
[.{$a''_0$}; \edge ;
[.\node(a1prime){$a''_1$}; \edge ;
[.{$a''_2$} \edge ;
[.{$a''_3$} \edge ;
[.{$a''_4$} 
]]]]]
\end{scope}

\begin{scope}[xshift=1.5cm,yshift=-.75cm,level distance=.8cm, sibling distance=0.5cm, font=\scriptsize]
\Tree
[.{$a'''_0$} \edge ;
[.{$a'''_1$} \edge ;
[.\node(a2prime){$a'''_2$}; \edge ;
[.{$a'''_3$} \edge ;
[.{$a'''_4$} 
]]]]]
\end{scope}

\draw[->] (a0prime) edge[out=-100,in=70,looseness=1.2] (a1) ;
\draw[->] (a1prime) edge[out=-80,in=100,looseness=1.2] (a2) ;
\draw[->] (a2prime) edge[out=-100,in=70,looseness=1.2] (a3) ;

\end{tikzpicture}
\\ \\
Input (uni-relational) & $\mathcal{F}_{a_0}$ single structure & $\mathcal{F}^*_{a_0}$ single structure \\
structure $(A,a_0)$
\end{tabular}
\caption{First example illustrating Definitions~\ref{def:F} and~\ref{def:F*}.}
\label{fig:example-frontier-1}
\end{figure}
  
\begin{figure}
\begin{tabular}{ccc}
\\
\begin{tikzpicture}[level distance=1.25cm,sibling distance=1cm, ->, baseline={(a)}]
\Tree
[.\node(a)[draw,circle]{$a$};
    \edge node[auto=right] {$R$};
    [.{$b$} 
       \edge node[auto=right] {$R$};
       [.{$c$} ]
       \edge node[auto=left] {$S$};
       [.{$d$} ]
    ]
]
\end{tikzpicture}
~~~~~~&~~~~~~
\begin{tikzpicture}[level distance=1.25cm,sibling distance=1cm, ->, baseline={(a)}]
\Tree
[.\node(a)[draw,circle]{$a$};
    \edge node[auto=right] {$R$};
    [.{$b$} 
       \edge node[auto=left] {$R$};
       [.{$c$} ]
        ]
    \edge node[auto=left] {$R$};
    [.{$b'$} 
        \edge node[auto=right] {$S$};
        [.{$d$} ]
        ]
]
\end{tikzpicture}
~~~~~~&~~~~~~
\begin{tikzpicture}[level distance=1.5cm,sibling distance=1.25cm, ->, baseline={(a)}]
\Tree
[.\node(a)[draw,circle]{$a$};
    \edge[thick] node[auto=right] {$R$};
    [.\node(b){$b$} ;
       \edge[thick] node[auto=left] {$R$};
       [.\node(c){$c$} ; ]
        ]
    \edge[thick] node[auto=left] {$R$};
    [.\node(bprime){$b'$} ;
        \edge[thick] node[auto=right] {$S$};
        [.\node(d){$d$} ; ]
        ]
]

\begin{scope}[xshift=-3cm,yshift=.25cm,level distance=1cm, sibling distance=0.5cm, font=\scriptsize]
\Tree [.\node(a1){} ; \edge node[auto=right] {$R$} ; [.{} \edge node[auto=right,pos=.8] {$R$} ; [.{} ] \edge node[auto=left,pos=.8] {$S$} ; [.{} ] ] ]
\end{scope}

\begin{scope}[xshift=-3cm,yshift=-2cm,level distance=1cm, sibling distance=0.5cm, font=\scriptsize]
\Tree [.{} \edge node[auto=right] {$R$} ; [.\node(b1){} ; \edge node[auto=right,pos=.8] {$R$} ; [.{} ] \edge node[auto=left,pos=.8] {$S$} ; [.{} ] ] ]
\end{scope}

\begin{scope}[xshift=3cm,yshift=.25cm,level distance=1cm, sibling distance=0.5cm, font=\scriptsize]
\Tree [.\node(a2){} ; \edge node[auto=left] {$R$} ; [.{} \edge node[auto=right,pos=.8] {$R$} ; [.{} ] \edge node[auto=left,pos=.8] {$S$} ; [.{} ] ] ]
\end{scope}

\begin{scope}[xshift=3cm,yshift=-2cm,level distance=1cm, sibling distance=0.5cm, font=\scriptsize]
\Tree [.{} \edge node[auto=left] {$R$} ; [.\node(b2){} ; \edge node[auto=right, pos=.8] {$R$} ; [.{} ] \edge node[auto=left,pos=.8] {$S$} ; [.{} ] ] ]
\end{scope}

\draw[->] (a1) edge[out=-60,in=120,looseness=1.2] (b) node[pos=.5,xshift=-2cm,yshift=-.2cm] {$R$};
\draw[->] (a2) edge[out=240,in=60,looseness=1.2] (bprime) node[pos=.5,xshift=2cm,yshift=-.2cm] {$R$};
\draw[->] (b1) edge[out=-60,in=140,looseness=1.2]  (c) node[pos=.5,xshift=-2cm,yshift=-2.4cm] {$R$};
\draw[->] (b2) edge[out=220,in=60,looseness=1.2]  (d) node[pos=.5,xshift=2cm,yshift=-2.4cm] {$S$};

\end{tikzpicture}
\\
Input structure & $\mathcal{F}_{a}$ single structure & $\mathcal{F}^*_{a}$ single structure \\
$(A,a)$
\end{tabular}
\caption{Second example illustrating Definitions~\ref{def:F} and~\ref{def:F*}.}
\label{fig:example-frontier-2}
\end{figure}

The examples in Figures~\ref{fig:example-frontier-1} and~\ref{fig:example-frontier-2} 
illustrate the construction of $\mathcal{F}_a$ and $\mathcal{F}^*_a$. In these figures, for clarity, the distinguished element is marked by a circle. In Figure~\ref{fig:example-frontier-1}, all edges represent the same binary relation $R$, and edge labels are omitted for the sake of readability.
In both examples, it happens (coincidentally) that the frontier consists of a single structure. 

Definition~\ref{def:radial-ext} and \ref{lem:radial-ext} further down may provide additional intuition on the $(\cdot)^*$ operation used in Definition~\ref{def:F*}.

 \begin{theorem} \label{thm:acyclic-unary}
Let $(A,a_0)$ be a finite structure that is \emph{core}, \emph{c-connected} and \emph{acyclic}, 
let $a$ be any node of $A$, and let $\mathcal{F}^*_a$ be as defined above.
\begin{enumerate}
\item
 For each $F^*\in \mathcal{F}^*_a$, $(F^*,a)\to (A,a)$. 
\item
 For each $F^*\in \mathcal{F}^*_a$, $(A|a,a)\not\to (F^*,a)$.
 \item 
 Let $(B,b)$ be  acyclic and c-connected. If $(B,b)\to (A,a)$ and $(A|a,a)\not\to (B,b)$, then $(B,b)$ maps homomorphically to $(F,a)$ for some structure $F\in\mathcal{F}^*_a$.  
\end{enumerate}
In particular (since $A|a_0=A$),  $\mathcal{F}^*_{a_0}$ is a frontier for $(A,a_0)$ w.r.t.~the class of c-connected, acyclic structures with one distinguished element.
 \end{theorem}
 
\begin{proof}
 Item 1 follows immediately from the construction: the natural projection from $(F^*,a)$ to $(A,a)$  is a homomorphism.
 
 For item 2, we proceed by induction on $rank(a)$.  
 Item 2 holds true, trivially, when $rank(a)=0$, as $\mathcal{F}^*_a$ is empty in this case. 
 For the inductive case, now, assume that $rank(a)>0$. In this case, by definition, we know that $F\in \mathcal{F}_a$ is of the form $H_i$ for some child $a'_i$ of $a$.
 In other words, $F$ was obtained from $A|a$ by removing the subtree $A|a'_i$, and (provided $rank(a'_i)>0$), adding a fresh isomorphic copy of each structure in $\mathcal{F}_{a'_i}$ joining $a$ with the newly created copy of $a'_i$ with a $S_i$-edge, where $S_i$ is the label of the edge $a\to a'_i$ in the original structure.

 Consider the homomorphism $h:(F^*,a)\to (A,a)$ given by item 1.
  Suppose for the sake of a contradiction that there were also a homomorphism 
$h': (A|a,a)\to (F^*,a)$.  Let $h'(a'_i)=b$. Towards our contradiction, we perform a
case distinction on $b$.
%
%
Clearly, $b$ must be one of the neighbours of $a$ in $F^*$. By construction, these neighbours of $a$ in $F^*$ are:
(i)
the children $a'_1, \ldots, a'_{i-1}, a'_{i+1}, \ldots, a'_n$ in $F$; (ii) the copies of $a'_i$ belonging to 
isomorphic copies of structures in $\mathcal{F}_{a'_i}$; and, provided $a\neq a_0$, (iii) the parent, $p$, of $a$ in $A$, as well as the copy of $p$ introduced in Definition~\ref{def:F*}. 

Let us first consider case (i) and (iii). In both cases, let $h''$ be the composition of $h'$ and $h$. Then $h''$ is a homomorphism from $(A|a, a)$ to $(A,a)$ whose range omits $a'_i$. 
Furthermore, $h''$ extends straightforwardly to a non-injective endomorphism on $(A,a_0)$
by mapping all elements outside $A|a$ to themselves. This is a contradiction with the fact that $(A,a_0)$ is a core.

Finally consider case (ii). By a similar argument as the above, no element $c$ from $A|a'_i$ can be mapped by $h'$ to $a$. For, in this case, the composition $h''$ of $h'$ and $h$ would be a homomorphism from $(A|a, a)$ to $(A,a)$ whose range excludes $c$, and $h''$ could then be extended to a non-injective endomorphism on $(A,a_0)$, contradicting the fact that $(A,a_0)$ is a core.
It follows that the restriction of $h'$ to $A|a'_i$ must be such that its range is entirely contained in 
$(F'^*,a'_i)$ for some $F'$ in $\mathcal{F}^*_{a'_i}$, i.e., it 
defines a homomorphism from $A|a'$ to $(F'^*,a'_i)$, which contradicts the inductive hypothesis.


Item 3 is proved by induction on $rank(a)$. Again, item 3 holds true, trivially, when $rank(a)=0$. Note that when $rank(a)=0$, $A|a$ is a single-node structure without any relations. Therefore, it is impossible that $(A|a,a)\not\to (B,b)$. 

Now, consider the inductive case with $rank(a)>0$.  
Since $(A|a,a)\not\to (B,b)$, it must be the case that, for some child $a'_i$ of $a$ there is no element $b'$ in $B$ such that  $(A|a'_i,a'_i)\to (B,b')$ and $S_i(b,b')$ holds in $B$ where $S_i$ is the label of the edge joining $a$ and $a'_i$. 
Let $F=H_i\in \mathcal{F}_a$ be the corresponding structure as constructed in Definition~\ref{def:F}. 

Let $h$ be the homomorphism from $(B,b)$ to $(A,a)$ given by the hypothesis. We shall construct a homomorphism $h'$ from  $(B,b)$ to $(F^*,a)$. We have $h'(b)=a$. Note that, since $B$ is acyclic and c-connected, we can similarly regard $B$ as an ordered tree rooted at $b$.  Then, for each child, $b'$, of $b$, we define $h'$ on $B|b'$ as follows. If $h(b')=a'_i$, then since $(A|a'_i,a'_i)$ has no homomorphism to $(B,b')$,  we apply the inductive hypothesis to define $h'$ on $B|b'$. If $h(b')$ is the parent of $a$ then we map $B|b'$ entirely to the isomorphic copy of $A$ attached to $a$ introduced 
in Definition \ref{def:F*}. If $h(b')$ is some child of $a$ different than $a'_i$ then we define $h'$ on $B|b'$ starting at $b'$ and by increasing depth as in the homomorphism, $h$, from $(B,b)$ to $(A,a)$ until we find some edge $c\to d$ in $B|b'$ such that its $h$-image is not an edge in $A$. When we found such edge $c\to d$, then we define
$h'(d)$ to be the copy of $h(d)$ attached to $h(c)$ in $F^*$ according to Definition \ref{def:F*} and we extend $h'$ to the rest of elements in
$B|d$ mapping them as well to the same isomorphic copy.
\end{proof}


\begin{definition}
The \emph{size} of a structure $A$, denoted by $size(A)$, is the number of facts. 
The \emph{total size} of a set of structures is $totalsize(\mathcal{A}) = \Sigma_{A\in \mathcal{A}} (size(A)+1)$ \end{definition}

(The definition of total size is conveniently chosen so that the total size of a set of tree-shaped structures is equal to the size of a tree-shaped structure consisting of the given forest with an additional root connected to the root of each original tree.)

\begin{theorem}[Our construction is polynomial]
$totalsize(\mathcal{F}^*_{a_0}) = O(size(A)^4)$.
\end{theorem}

\begin{proof} We will show that  $totalsize(\mathcal{F}_{a_0}) \leq size(A)^3$. The proposition then follows immediately, by  construction of $\mathcal{F}^*$. More precisely, we show that, for all elements $a$, $totalsize(\mathcal{F}_a)\leq size(A|a)^3$. 
The  claim is proved by induction on $rank(a)$.

If $rank(a)=0$, then $totalsize(\mathcal{F}_a) = 0$.
If $rank(a)>0$, it follows from the construction of $\mathcal{F}_a$ that 
\[
\begin{split}
totalsize(\mathcal{F}_a) &\leq n \cdot size(A|a) + \Sigma_{i=1}^n totalsize(\mathcal{F}_{a'_i}) \\
& \leq size(A|a)^2 + \Sigma_{i=1}^n size(A|a'_i)^3 \\
& \leq size(A|a) \cdot \big(size(A|a) + \Sigma_{i=1}^n size(A|a'_i)^2\big) \\
& \leq size(A|a) \cdot \big(size(A|a) + (size(A|a)-1)^2\big) \\
& \leq size(A|a) \cdot size(A|a)^2 \\
& \leq  size(A|a)^3
\end{split}
\]
where $n$ is the number of successors of $a$. Note that we are using here the fact that 
$\Sigma_{i=1}^n size(A|a'_i) \leq size(A|a)-1$, and the general fact that $\Sigma_{i=1}^n x_i^2 \leq (\Sigma_{i=1}^n x_i)^2$.
\end{proof}

\subsection*{Extending the result to $k\geq 1$}

We now extend the result to $k\geq 1$. For the time being, we still assume that the schema consists of binary relations only. 

\begin{definition}[Skeleton and offshoots]
Let $(A,\textbf{a})$ be a c-connected, acyclic structure with at least one distinguished element.
A \emph{skeleton node} of $(A,\textbf{a})$ (with $\textbf{a}=a_1, \ldots, a_k$) is any element $s$ that is either a distinguished element (i.e., $s \in \{a_1, \ldots, a_k\}$) or such that $s$ lies on a minimal path between two distinguished elements. We denote by $\skeleton(A,\textbf{a})$ the substructure of 
$(A,\textbf{a})$ (with the same distinguished elements), consisting of the skeleton nodes only. It is easy to see that, in a c-connected and acyclic structure, for every element $c$ there is a unique skeleton node $s$ that is closest to $c$, i.e., such that $c$ is connected to $s$ and such that every path from $c$ to any other skeleton node passes through $s$. In this case, we say that $c$ is \emph{affiliated with} the skeleton node $s$.
For any skeleton node $s$ of $(A,\textbf{a})$, the \emph{offshoot} of $s$ in $(A,\textbf{a})$, which we will denote by $(A|s,s)$,
\footnote{We acknowledge that this notation is slightly misleading as it does not reflect the dependence of $(A|s,s)$ on the distinguished elements $\textbf{a}$.}
is the structure $(A',s)$ where $A'$ is the substructure of $A$ consisting of all elements affiliated with $s$ (including $s$ itself), and where $s$ is the only distinguished element.
\end{definition}

Thus, every c-connected and acyclic structure can be thought of as consisting of a skeleton with offshoots, similar to the example depicted schematically in Figure~\ref{fig:skeleton}.

\begin{figure}
    \centering
\usetikzlibrary{shapes.geometric}
\tikzset{
triangle1/.style={
  draw,shape border uses incircle,
  isosceles triangle,shape border rotate=55,yshift=0cm,xshift=-.2cm},
triangle2/.style={
  draw,shape border uses incircle,
  isosceles triangle,shape border rotate=125,yshift=0cm,xshift=.2cm},
triangle3/.style={
  draw,shape border uses incircle,
  isosceles triangle,shape border rotate=90,yshift=0cm},
}
    \begin{tikzpicture}
    \node[draw,circle,fill=black,label={$a_1$}] at (0, 0)     (a1) {}; \node[triangle1] at (-0.1,-0.5) {};
    \node[draw,circle,fill=black,label=below right:{$a_2$}] at (1.5, -2)  (a2) {}; \node[triangle1] at (1.4,-2.5) {};
    \node[draw,circle,fill=black,label={$a_3$}] at (5, -2)    (a3) {}; \node[triangle2] at (5.1,-2.5) {};
    \node[draw,circle,fill=black,label={$a_4$}] at (3.5, 0)   (a4) {}; \node[triangle2] at (3.6,-0.5) {};
    \node[draw,circle,fill=black,label={$a_5$}] at (6, -1)    (a5) {}; \node[triangle2] at (6.1,-1.5) {};
    \node[draw,circle,fill=black] at (1, -.7)   (b1) {};  \node[triangle1] at (0.9,-1.2) {};
    \node[draw,circle,fill=black] at (2, -1)    (b2) {};  \node[triangle2] at (2.1,-1.5) {};
    \node[draw,circle,fill=black] at (3, -1)    (b3) {};  \node[triangle3] at (3,-1.6) {};
    \node[draw,circle,fill=black] at (4, -1.3)  (b4) {};  \node[triangle1] at (3.9,-1.8) {};

    \draw[thick] (a1) -- (b1) -- (b2) -- (b3) -- (b4) -- (a3);
    \draw[thick] (a2) -- (b2);
    \draw[thick] (b3) -- (a4);
    \end{tikzpicture}

    \caption{Schematic depiction of an example c-connected acyclic structure with five distinguished elements. The skeleton consists of the black nodes (which are the elements of the structure that are either distinguished or lie on a shortest path connecting two distinguished elements), while the triangles depict offshoots.}
    \label{fig:skeleton}
\end{figure}

Our frontier construction for structures with multiple distinguished elements
will make use of two operations that we now define.

\newcommand{\splitting}{\textrm{\upshape split}}
\begin{definition}[Splitting]
Consider a structure $(A,\textbf{a})$ with $\textbf{a}=a_1, \ldots, a_k$.
Let $(X,Y)$ be a proper partition of $\{1, \ldots, k\}$, that is,
    $X$ and $Y$ are disjoint non-empty sets such that $X\cup Y = \{1, \ldots, k\}$.
    We will denote by $\splitting_{(X,Y)}(A,\textbf{a})$ the structure $(A', \textbf{a}')$ where
    \begin{itemize}
        \item $A'=A^{(1)}\uplus A^{(2)}$ is a disjoint union of two isomorphic copies of $A$. 
          For each element $a$ of $A$, we will denote its two copies in $A'$ 
          by $a^{(1)}$ and $a^{(2)}$, respectively. 
          \item For $\textbf{a}=a_1, \ldots, a_n$, we set $\textbf{a}'=a_1^{(j_1)}, \ldots, a_n^{(j_n)}$ where $j_i= 1$ if $i\in X$ and $j_i=2$ otherwise.
    \end{itemize}
\end{definition}

For the next definition, we need to introduce some auxiliary notation.
If $a, b$ are elements belonging to the same connected component of some structure $A$, 
we will denote by $dist(a,b)$ the length of the 
shortest path from $a$ to $b$ in the incidence graph of $A$ (where the
length may be counted by the number of facts on the path). 
If $a,b,c$ are elements that all belong to the same connected component of $A$
we will write $a<_c b$ if $dist(a,c) < dist(b,c)$ (``$a$ is closer to $c$ than $b$ is'').

\newcommand{\radialext}{\textrm{\upshape radial-extension}}
\begin{definition}[Radial extension]\label{def:radial-ext}
Let $h:(A,s)\to (B,t)$ be a homomorphism between acyclic structures with one distinguished element. We denote by $\radialext_{h:(A,s)\to (B,t)}(A)$ the structure obtained from $A$ as follows: for each fact $f$ containing elements $a_1, a_2$ with $a_1 <_s a_2$, we add a disjoint isomorphic copy $\widetilde{B}$ of $B$ (without 
distinguished elements), and we add a connecting fact that is a copy of $f$ in which
$a_1$ is replaced by the isomorphic copy of $h(a_1)$ in $\widetilde{B}$.
\end{definition}

Note that the way in which $\mathcal{F}^*_a$ was constructed from $\mathcal{F}_a$ in Definition~\ref{def:F*}, is a concrete instance of the radial-extension operation.

In the statement of the next lemma, we write $(A,\textbf{a}, s)$,
where $(A,\textbf{a})$ is a structure with $k$ distinguished elements, 
to denote the corresponding structure with $k+1$ distinguished elements.

\begin{lemma}\label{lem:radial-ext}
Let $(A,\textbf{a}, s)$ and $(B, \textbf{b}, t)$ be acyclic structures with $(B,\textbf{b}, t)\stricthom (A,\textbf{a}, s)$, 
and such that $(A,\textbf{a}, s)$ is a core.
Let $B' = \radialext_{h:(B,t)\to (A,s)}(B)$ for some $h:(B,\textbf{b}, t)\to (A,\textbf{a}, s)$. 
Then $(B,\textbf{b}, t) \to (B',\textbf{b}, t)\stricthom(A,\textbf{a}, s)$.
\end{lemma}

\begin{proof}
 
 Clearly, $(B',\textbf{b}, t)$ extends $(B,\textbf{b}, t)$. It is also clear from the construction that the map $h$ naturally extends to a homomorphism $g:(B',\textbf{b},t)\to (A,\textbf{a},s)$ (using the natural projection for the additional elements of $B'$).
 
 Let us show now that $(A,\textbf{a}, s)$ does not map homomorphically to  $(B',\textbf{b}, t)$.  For the sake of a contradiction, assume that $h':(A,\textbf{a}, s)\to (B',\textbf{b}, t)$.
We can assume from Lemma~\ref{lem:core-injective} that the composition of $g$ and $h'$ is injective.
Since
$(A,\textbf{a}, s)\not\to (B, \textbf{b}, t)$, $h'$ must map some element $c$ of $A$ to 
 $h'(c)=d$, where $d$ is an element of $B'$ that does not belong to $B$. Then $d$ belongs to 
 an isomorphic copy of $A$ that was added for some fact $f$ of $B$, where $f$ contains
 elements $b_1, b_2$ with $b_1 <_t b_2$. Let $b_1'$ be the copy of $h(b_1)$ belonging to
 the respective isomorphic copy of $A$. Since $h'(s)=t$ it then follows that the image according to $h'$ of the nodes in the shortest path connecting $c$ with $s$ must necessarily contain all nodes  in the shortest path (in $B'$) connecting $d$ and $t$. Note that the path connecting $d$ and $t$ must contain
 both $b'_1$ and $b_1$.  Since $g(b'_1)=g(b_1)$
 it follows that the composition of $g$ and $h'$ is non-injective, a contradiction.
\end{proof}

Now, let $(A,\textbf{a})$ be c-connected and acyclic.
We may also assume that $(A,\textbf{a})$ is a core. 
We define $\mathfrak{F}$ to be the set of all structures obtained as follows:%
\footnote{Note that $\mathfrak{F}$ depends on $(A,\textbf{a})$, even though we did not
make this dependence explicit in our notation here.}
\begin{enumerate}
    \item For each proper partition $(X,Y)$ of the distinguished elements, we add
    $\splitting_{(X,Y)}(A,\textbf{a})$ to $\mathfrak{F}$, provided there are $a_i\in X$ and $a_j\in Y$ such that $a_i$ and $a_j$ are connected in $A$. (Note that, by construction,
    the corresponding distinguished elements $a_i^{(1)}$ and $a_j^{(2)}$ are disconnected in $\splitting_{(X,Y)}(A,\textbf{a})$.)
    
    \item For each proper partition $(X,Y)$ of the distinguished elements, and for each fact $R(c,d)$ of $A$, let $(B,\textbf{b})$ be the structure obtained by extending $\splitting_{(X,Y)}(A,\textbf{a})$ with the fact
    $R(c^{(1)}, d^{(2)})$. We add $(B,\textbf{b})$ to $\mathfrak{F}$, provided that there
    exists distinguished elements $a_i\in X$ and $a_j\in Y$, such that $a_i$ and $a_j$ are connected in $A$ by a path of some length $n$, but there is no path of length $n$ connecting the corresponding distinguished elements $a_i^{(1)}$ and $a_j^{(2)}$ in $(B,\textbf{b})$.
    
\item For each skeleton node $s$, and for every $(F,s)$ in the frontier of the offshot $(A|s,s)$ (as in Theorem~\ref{thm:acyclic-unary}) we add to $\mathfrak{F}$ the structure obtained from $(A,\textbf{a})$ by: (i)     removing the offshoot $(A|s,s)$ and replacing it $(F,s)$, resulting in a new structure $(A',\textbf{a})$, and (ii)      taking $B = \radialext_{h:(A',s)\to (A,s)}(A')$, where $h$ is the homomorphism from $(F,s)\to (A|s,s)$, extended to the entire structure $A'$ by mapping every element outside the offshoot to itself.

\end{enumerate}
The intuition behind this construction is that (1)--(3) capture different ways of homomorphically weakening an acyclic structure with designated elements: in (1), we take two connected distinguished elements and force them to become disconnected. In (2), we increase the length of the shortest path between two distinguished elements, by forcing the minimal path to go through a specified edge. In (3), we make no change to the skeleton of the structure but we weaken one of the offshoots. 

\begin{proposition} If $(A,\textbf{a})$ is c-connected, acyclic, and is a core, then
the set of structures $\mathfrak{F}$ constructed above is a frontier for $(A,\textbf{a})$ w.r.t. the class of c-connected acyclic structures.
\end{proposition}

\begin{proof}
It is clear that each structure in $\mathfrak{F}$ homomorphically maps to $(A,\textbf{a})$.

We also claim that there is \emph{no} homomorphism from $(A,\textbf{a})$ to any structure $(B,\textbf{b})\in \mathfrak{F}$. For (1) and (2) this follows immediately from the construction. For (3), the argument is as follows: 
let $s$ be the skeleton node whose offshoot was replaced, and let $(F,s)\in \mathfrak{F}^*_s$ be the frontier structure was used as the replacement of $(A|s,s)$.
Suppose, for the sake of a contradiction, that $(A,\textbf{a})\to (B,\textbf{b})$. By Lemma~\ref{lem:radial-ext}, then, there is already a homomorphism
$h:(A,\textbf{a},s)\to (A',\textbf{a},s)$.
By Lemma~\ref{lem:core-injective}, we may assume that the composition of  $h$ with the natural homomorphism from $h':(A',\textbf{a},s)$ to $(A,\textbf{a},s)$ is the identity function on $A$. 
Since $(A|s,s)\not\to(F,s)$ is follows that there exists some element $a$ in $A|s$ such that $h(a)$ does not belong to $(F,s)$. However, this contradicts the fact that $h'(h(a))$ is the identity. 

It remains to establish the last property of frontiers: let $(C,\textbf{c})$  be any c-connected and acyclic structure such that there is a homomorphism $h:(C,\textbf{c})\to (A,\textbf{a})$ and such that $(A,\textbf{a})\not\to (C,\textbf{c})$.

We can distinguish three cases:

The first case is where $(C,\textbf{c})$ differs from $(A,\textbf{a})$ in the 
number of connected components. Since both structures are c-connected and 
$(C,\textbf{c})\to (A,\textbf{a})$, this can only happen if there are
distinguished elements $c_i, c_j$ that belong to different components in 
$(C,\textbf{c})$ but such that $h(c_i), h(c_j)$
are connected in $(A,\textbf{a})$. In this case, we use (1) above.
Specifically, let $X = \{h(c_k) \mid \text{$c_k\in\textbf{c}$ is connected to $c_i$}\}$, and let $Y = \{\textbf{a}\}\setminus X$. Note that
$h(c_j)\in Y$. Then it is easy to see that $(C,\textbf{c})\to \splitting_{(X,Y)}(A,\textbf{a})$.

The second case is there there exists distinguished elements $c_i,c_j$ connected in $C$
such that the image of the path $s_1,\dots,s_m$ in $C$ connecting $c_i$ and $c_j$ is not injective. It follows from the acyclicity of $A$ that there exists $\ell$ such that $h(s_{\ell})=h(s_{\ell+2})$. 
Let $f$ be the fact in $C$ joining $s_{\ell}$ and $s_{\ell+1}$ and let $Z$ be the set containing all
elements in $C$ that remain connected to $c_i$ after removing fact $f$. Then consider the structure $(B,\textbf{b})$ obtained as in (2) above by setting $X$ to be the set of distinguished elements in $h(Z)$, $Y$ to be the rest of distinguished elements in $C$, and the fact of $A$ used in the construction to be the $h$-image of $f$. If we let $a_i=h(c_i)$ and $a_j=h(c_j)$ it follows directly from the construction that the distance of $a_i^{(1)}$ and $a_j^{(2)}$ in $(B,\textbf{b})$ is larger than the distance of $a_i$ and $a_j$ in $A$, and hence $(B,\textbf{b})\in \mathfrak{F}$. Finally, it follows easily that the map $h'$ sending every element $c\in C$ to $h(c)^{(1)}$ if $c\in Z$ and to $h(c)^{(2)}$ otherwise defines a homomorphism
from $(C,\textbf{c})$ to $(B,\textbf{b})$. 

The third case is where none of the previous two cases hold. We first note that if the second case does not hold then it must be the case that the restriction of $h$ to each connected component of
$\skeleton(C,\textbf{c})$ must be injective. Furthermore, since $(C,\textbf{c})$ and $(A,\textbf{a})$
have the same number of connected components it follows that $h$ maps $\skeleton(C,\textbf{c})$ isomorphically to $\skeleton(A,\textbf{a})$. It follows that there exists some node $s$ in the skeleton of $C$ such that the offshoot $(A|h(s),h(s))$ does \emph{not} homomorphically map to $(C,s)$, since otherwise, $(A,s)\to (C,s)$. Consider the offshoot $(C|s,s)$ of $C$ and let $D$ be the maximal substructure of $C$ that contains $s$, is connected, and satisfies the property that $h(D)$ is contained in $A|h(s)$ and that no other element, besides $s$, is mapped by $h$ to $h(s)$. 
Since $(D,s)\stricthom (A|h(s),h(s))$ it follows that there exists some homomorphism $g$ from $(D,s)$ to some structure $(F,h(s))$ in the frontier of $(A|h(s),h(s))$. Now consider the structure $(B,\textbf{b})$ in 
$\mathfrak{F}$ produced in step (3) above for $h(s)$ and $(F,h(s))$. We shall construct a homomorphism $h'$ from $(C,\textbf{c})$ to $(B,\textbf{b})$. 

Let $(A',\textbf{a})$ as constructed in step (3) and let $(E,\textbf{c})$ be the maximal $c$-connected substructure of $(C,\textbf{c})$ containing $\skeleton(C,\textbf{c})$ such that no other element in $E$ besides $s$ is mapped by $h$ to $h(s)$. Note that $E$ contains $D$.
We shall start by defining $h'$ on $E$ so that $h'$ defines a homomorphism from  $(E,\textbf{c})$ to $(A',\textbf{a})$.
If $e\in D$ then we define $h'(d)$ to be $g(d)$.
If $e$ does not belong to $D$ then it follows that $h(e)$ cannot be
in the offshot $A|h(s)$. In this case we define $h'(e)$ to be $h(e)$.
Clearly $h'$ defines as well a homomorphism from $(E,\textbf{c})$ to $(B,\textbf{a})$ (since $B$ contains a copy of $A'$). It only remains to extend $h'$ to all the elements in $C$. For every maximal connected substructure $F$ of $C$ not containing any element in $E$ we extend $h'$ to $F$ as follows. Since $E$ contains $\skeleton(C,\textbf{c})$
and $(C,\textbf{c})$ is $c$-connected it follows that there is some fact in $C$ joining elements $e\in E$ and $f\in F$. By the maximality of $E$ it follows that $h(f)=h(s)$. Necessarily $h(f) <_{h(s)} h(e)$ and, consequently, we
can define $h'$ so that $f$ and every other element in $F$ is mapped according to $h$ in the copy of $A$ in $B$ introduced due to $h(f) <_{h(s)} h(e)$. 

Incidentally, it may be worth noting that we did not even make full use of the radial extension: for the purpose of this proof, it would have sufficed in case (3) to use a restricted version of radial extension where a copy of the original structure is added only for every fact 
\emph{connected with $s$}.
\end{proof}

\subsection*{Lifting the  restriction to binary relations}
This concludes the proof for the case with binary relations only.
We now show how to lift the result to schemas containing relations of arbitrary arity.

For a schema $\mathcal{S}$, let $\mathcal{S}^*$ be the schema containing for each $n$-ary relation $R\in \mathcal{S}$, $n$ binary relations $R_1, \ldots, R_n$.
For any structure $A$ over schema $\mathcal{S}$, let $A^*$ be the structure over schema $\mathcal{S}^*$ whose domain consists of all elements in the domain of $A$ as well as all facts of $A$, and containing all facts of the form $R_i(b,f)$ where $f$ is a fact of $A$, of the form $R(\textbf{a})$ with $a_i = b$. Intuitively, we can think of $A^*$ as a bipartite encoding of the structure $A$. Conversely, we associate to every structure $B$ over the schema $\mathcal{S}^*$ a corresponding structure $B_*$ over the original schema $\mathcal{S}$, namely the structure whose domain is the same as that of $B$ and containing all facts of the form $R(a_1, \ldots, a_n)$ for which it is the case that $B$ satisfies $\exists y \bigwedge_{i=1\ldots n} R_i(a_i, y)$. 
Note that $(A^*)_* = A$ but $(B_*)^*$ need not be isomorphic to $B$.

\begin{lemma}\label{lem:binary-reduction}
For all structures $A, A'$ over schema $\mathcal{S}$ and structures $B$ over schema $\mathcal{S}^*$:

\begin{enumerate}
\item If $(B,\textbf{b})\to (A^*,\textbf{a})$ then $(B_*,\textbf{b})\to (A,\textbf{a})$.
\item $(A,\textbf{a})\to (B_*,\textbf{b})$ iff $(A^*,\textbf{a})\to (B,\textbf{b})$.
\item $(A,\textbf{a})\to (A',\textbf{a'})$ iff $(A^*,\textbf{a})\to ({A'}^*,\textbf{a'})$.
\item If $(A,\textbf{a})$ is core, c-connected, and acyclic, then so is $(A^*,\textbf{a})$
\item If $(B,\textbf{b})$ is acyclic, then so is $(B_*,\textbf{b})$. 
\end{enumerate}
\end{lemma}

\begin{proof} 
1. Let $h:(B,\textbf{b})\to (A^*,\textbf{a})$. 
   It is easy to see that, 
   for each element $b$ of $B_*$ that participates in at least one fact, it must
   be the case that $h(b)$ belongs to the domain of $A$. Note that elements of $B_*$ that do not participate in any fact can be ignored, as they can be mapped to an
   arbitrary element of $A$. Finally, it is clear from the constructions that,
   whenever $R(b_1, \ldots, b_n)$ holds true in $B_*$, then $R(h(b_1), \ldots, h(b_n))$ holds true in $A$. 
   
2. Suppose $h: (A,\textbf{a})\to (B_*,\textbf{b})$. We can extend $h$ to the entire domain of $A^*$ as follows: let $f$ be any fact of $A$ of the form 
  $R(a_1, \ldots, a_n)$. Since $B_*$ satisfies $R(h(a_1), \ldots, h(a_n))$, this means that
 $B$ must satisfy $\exists y \bigwedge_{i=1\ldots n} R_i(h(a_i), y)$. 
  Choose any such $y$ as the image of the fact $f$. Doing this for each fact,
  we obtain a mapping $h'$ that extends $h$ to the entire domain of $A^*$. Moreover,
  whenever $R_i(a_i,f)$ holds in $A^*$, then, by construction, $R_i(h'(a_i), h'(f))$ holds true in $B$. In other words $h'$ is a homomorphism from $(A^*,\textbf{a})$ to  $(B,\textbf{b})$
  
  Conversely, suppose $h:(A^*,\textbf{a})\to (B,\textbf{b})$. Let $h'$ be the restriction of $h$ to elements of $A$. We claim that the mapping $h': (A,\textbf{a})\to (B_*,\textbf{b})$ is a homomorphism. Let $f = R(a_1, \ldots, a_n)$ be any fact of $A$. Then $B$ satisfies $R_i(h'(a_i),h(f))$ for all $i=1\ldots n$. 
  Therefore, by construction, $B_*$ satisfies $R(h'(a_1), \ldots, h'(a_n))$.

3. Every homomorphism $h:(A,\textbf{a})\to (A',\textbf{a'})$ natural extends to a homomorphism from $(A^*,\textbf{a})$ to $({A'}^*,\textbf{a'})$ by sending each fact
$R(a_1, \ldots, a_n)$ to $R(h(a_1), \ldots, h(a_n))$. Conversely, every homomorphism
$h:(A^*,\textbf{a})\to ({A'}^*,\textbf{a'})$, when restricted to elements of $A$, is clearly a homomorphism from $(A,\textbf{a})$ to $(A',\textbf{a'})$.

4. For acyclicity, this holds by definition (recall that acyclicity was defined by reference to the incidence graph, in the first place; and note that the
  incidence graph of $(A^*,a)$ is obtained from that of $(A,a)$ by subdividing every
  edge in two). Similarly, it is easy to see
that whenever $(A,\textbf{a})$ is c-connected, then so is $(A^*,\textbf{a})$. Finally, 
to show that core-ness is preserved, we proceed by contraposition: suppose
$(A^*,\textbf{a})$ is not a core, i.e., admits a proper endomorphism $h$. It is not 
hard to see that $h$ must map elements of $A$ to elements of $A$, and fact of $A$ to facts of $A$. Therefore, $h$ must either map two distinct elements of $A$ to the same element, or two distinct facts of $A$ to the same fact. However, even in the latter case, the only way that this can happen is if $h$ also maps two distinct elements to the same element. It follows that $h$ also induces a proper endomorphism of $A$.

5. By contraposition: suppose $(B_*,\textbf{b})$ is not acyclic. Take a minimal cycle in the incidence graph of $(B_*,\textbf{b})$. Each edge that is part of the 
cycle (being a fact of $B_*$), by construction of $B^*$, gives rise to a path of length 2 in $B$. Therefore, we obtain a cycle in $B$.
\end{proof}

Now, let $(A,\textbf{a})$ be any acyclic, c-connected structure over schema $\mathcal{S}$.
We may again assume that $(A,\textbf{a})$ is a core.
By Lemma~\ref{lem:binary-reduction}(4), $(A^*,\textbf{a})$ is also acyclic, c-connected, and core.
Let $F$ be an acyclic c-connected frontier for $(A^*,\textbf{a})$, let $F' = \{(B_*,\textbf{b})\mid (B,\textbf{b})\in F\}$.
We claim that $F'$ is a frontier for $(A,\textbf{a})$. 

First, note that each structure in $F'$ homomorphically maps to $(A,\textbf{a})$. This follows from Lemma~\ref{lem:binary-reduction}(1), because each $(B,\textbf{b})\in F$ homomorphically maps to $(A^*,\textbf{a})$. 
Second, note that $(A,a)$ does not homomorphically map to any structure in $F'$. Indeed, suppose $(A,\textbf{a})\to (B_*,\textbf{b})$ with $(B_*,\textbf{b})\in F'$. Then, by Lemma~\ref{lem:binary-reduction}(2), $(A^*,\textbf{a})\to (B,\textbf{b})$, which contradicts the fact that $F$ is a frontier for $(A^*,\textbf{a})$.
Finally, for the third property of frontiers, suppose $(C,\textbf{c})\to (A,\textbf{a})$ and $(A,\textbf{a})\not\to (C,\textbf{c})$. Then, by Lemma~\ref{lem:binary-reduction}(3), $(A^*,a)\to (C^*,\textbf{c})$ and $(C^*,\textbf{c})\not\to (A^*,\textbf{a})$. Therefore, since $F$ is a frontier for $(A^*,\textbf{a})$, we have that $(C^*,\textbf{c})\to (B,\textbf{b})$ for some $(B,\textbf{b})\in F$. It follows by Lemma~\ref{lem:binary-reduction}(2) that $(C,\textbf{c})\to (B_*,\textbf{b})$. Since $(B^*,\textbf{b})\in F'$, this shows that $(C,\textbf{c})$ maps to a structure in $F'$. 

Note that $F'$ consists of acyclic structures (by Lemma~\ref{lem:binary-reduction}(5)), but may contain structures that are not c-connected. However, this is easily addressed by the following observation, which follows from Proposition~\ref{prop:reach}: if a set of structures $F$ is a frontier for a structure $(A,\textbf{a})$ w.r.t.~a class of c-connected structures $\mathcal{C}$, then $\{(B, \textbf{b})^{\textrm{reach}}\mid (B, \textbf{b})\in F\}$ is also a frontier for $(A,\textbf{a})$ w.r.t.~$\mathcal{C}$.

This concludes the proof of Theorem~\ref{thm:frontier-closed}.

\section{Unique Characterizations for Conjunctive Queries}
\label{sec:characterizations}

In this section, we study the question of when a CQ is uniquely characterizable by a finite set of positive and/or negative examples. 

\begin{definition}[Data Examples, Fitting, Unique Characterizations]
Let $\mathcal{C}$ be a class of $k$-ary CQs over a schema
$\mathcal{S}$ (for some $k\geq 0$), and let $q$ be a $k$-ary query over $\mathcal{S}$.
\begin{enumerate}
    \item A \emph{data example} is a structure $(A,\textbf{a})$ over 
    schema $\mathcal{S}$ with $k$ distinguished elements.
    If $\textbf{a} \in q(A)$, we call $(A,\textbf{a})$ a 
    \emph{positive example} (for $q$), otherwise a
    \emph{negative example}.
\item Let $E^+, E^-$ be finite sets of data examples. We say that $q$ \emph{fits}
$(E^+,E^-)$ if every example in $E^+$ is a positive example for $q$ and 
every example in $E^-$ is a negative example for $q$. We say that
$(E^+,E^-)$ \emph{uniquely
characterizes $q$ w.r.t.~$\mathcal{C}$} if $q$ fits $(E^+,E^-)$ and every 
$q'\in\mathcal{C}$ that fits $(E^+,E^-)$ is logically equivalent to $q$.
\end{enumerate}
\end{definition}

It turns out that there is a precise correspondence between unique characterizations and frontiers. Recall that the canonical structure of a query $q$ is denoted by $\widehat{q}$.
Similarly, for any class of CQs 
$\mathcal{C}$, we will denote by $\widehat{\mathcal{C}}$ the class of structures
$\{\widehat{q}\mid q\in\mathcal{C}\}$.

\begin{proposition}[Frontiers vs Unique Characterizations]
\label{prop:frontiers-characterizations}
Fix a schema $\mathcal{S}$ and $k\geq 0$.
Let $q$ be any $k$-ary CQ over $\mathcal{S}$ and 
$\mathcal{C}$ a class of $k$-ary CQs over $\mathcal{S}$.
\begin{enumerate}
    \item 
If $F$ is a frontier for $\widehat{q}$ w.r.t.~$\widehat{\mathcal{C}}$,  then $(E^+ = \{\widehat{q}\},E^- = F)$ uniquely characterizes $q$ w.r.t.~$\mathcal{C}$.
\item
Conversely, if $(E^+,E^-)$ uniquely characterizes $q$ w.r.t.~$\mathcal{C}$, 
then 
$F = \{ \widehat{q} \times (B,\textbf{b}) \mid 
     (B,\textbf{b})\in E^-\}$ is a frontier for $\widehat{q}$ w.r.t.~$\widehat{\mathcal{C}}$.
\end{enumerate}
\end{proposition}

\begin{proof} 
1. Let $F$ be a frontier for $\widehat{q}$ w.r.t.~$\widehat{\mathcal{C}}$,
let $q'\in\mathcal{C}$ be a conjunctive query that fits $(E^+ = \{\widehat{q}\},E^- = F)$.
From the fact that the canonical structure $\widehat{q}$ is a positive example for $q'$, it follows that there is a homomorphism from $\widehat{q'}$ to $\widehat{q}$. Furthermore, since
all structures in the set $F$ are negative examples for $q'$, we know that $\widehat{q'}$ does not homomorphically map to any of these structures. 
Since $F$ is a frontier w.r.t.~$\widehat{\mathcal{C}}$ and $\widehat{q'}\in\widehat{\mathcal{C}}$, we can conclude that $\widehat{q}$
homomorphically maps to $\widehat{q'}$. Therefore, $\widehat{q}$ and $\widehat{q'}$ are homomorphically equivalent, which implies that $q$ and $q'$ are logically equivalent.

2.
Let $(E^+,E^-)$ uniquely characterize $q$ w.r.t, $\mathcal{C}$,
and let $F = \{ \widehat{q} \times (B,\textbf{b}) \mid 
     (B,\textbf{b})\in E^-\}$.
It follows from the basic properties of the direct product operation that
each structure in $F$ homomorphically maps to $\widehat{q}$.
Furthermore, if there were a homomorphism from $\widehat{q}$ to 
some structure $\widehat{q} \times (B,\textbf{b})\in F$, then 
there would be a homomorphism from $\widehat{q}$ to $(B,\textbf{b})$,
which would imply that $(B,\textbf{b})$ is a positive example for $q$, 
which we know is not the case. 
Finally, consider any $\widehat{q'}\in\widehat{\mathcal{C}}$ such that
$\widehat{q'}\to\widehat{q}$ and $\widehat{q}\not\to\widehat{q'}$.
This implies that $q$ and $q'$ are not logically equivalent, 
and hence, since $q'\in\mathcal{C}$, the two queries must disagree on some example in $E^+$ or $E^-$. However, it follows from the fact that $\widehat{q}\to\widehat{q}$, that all positive examples for $q$ are also positive examples for $q'$. Therefore, some structure $(B,\textbf{b})\in E^+$ must be a positive example for $q'$, that is, $\widehat{q'}\to(B,\textbf{b})$. 
It follows that $\widehat{q'}\to \widehat{q}\times (B,\textbf{b})$.
\end{proof}

Proposition~\ref{prop:frontiers-characterizations} allows us to take the results on frontiers from the previous section, and rephrase them in terms of unique characterizations. Incidentally, note that results in~\cite{AlexeCKT2011} imply an analogous relationship between \emph{finite dualities} and uniquely characterizing sets of examples for \emph{unions of conjunctive queries}.
We need two more lemmas. Recall that a structure $(A,\textbf{a})$ corresponds to a conjunctive
query only if every distinguished element occurs in at least one fact. 
Let us call such structures \emph{safe}. The following lemmas, essentially, allow us 
to ignore unsafe structures, thereby bridging the gap between structures and
CQs. 

\begin{lemma} \label{lem:characterisation-safe}
Let $q$ be a $k$-ary CQ over schema $\mathcal{S}$ and $\mathcal{C}$ a class of $k$-ary CQs over 
$\mathcal{S}$. If $q$ is uniquely characterized w.r.t.~$\mathcal{C}$ by 
positive and negative examples
$(E^+,E^-)$, then $E^+$ consists of safe structures and 
$q$ is uniquely characterized w.r.t.~$\mathcal{C}$ by 
$(E^+, \{(A,\textbf{a})\in E^-\mid \text{$(A,\textbf{a})$ is safe}\})$.
\end{lemma}

\begin{proof} 
It suffices to observe that if $(A,\textbf{a})$ is not safe, then 
$(A,\textbf{a})$ is a negative example for every conjunctive query, and therefore,
cannot meaningfully contribute to characterizing any given conjunctive query w.r.t.~a
class of CQs.
\end{proof}

\begin{lemma}\label{lem:frontier-safe}
A safe structure has a frontier w.r.t.~all structures if and only if it has a frontier w.r.t.~the class of all safe structures.
\end{lemma}

\begin{proof} 
The left-to-right direction is trivial. The right-to-left direction relies on a 
homomorphism duality argument of sorts: let $(A,\textbf{a})$ be any safe structure that has a frontier $F$ w.r.t.~the class of all safe structures. Let $\mathcal{S}$ be its schema and $k$ the number of distinguished elements.
For every non-empty set $S\subseteq \{1, \ldots, k\}$, we will
denote by $(A_S,\textbf{a}_S)$ the structure with two elements, denoted $b$ and $c$, 
 that contains all possible facts involving only $c$ and no other facts; and 
 $\textbf{a}_S$ is the tuple $a_1, \ldots, a_k$, where $a_i=b$ if $i\in S$, and $a_i=c$ otherwise. Note that, by construction, none of these structures is safe.
It is not hard to see that a structure is unsafe if and only if it admits a homomorphism
to a structure in the set $G=\{(A_S,\textbf{a}_S)\mid 
\text{$S\subseteq \{1, \ldots, k\}$ is non-empty}\}$. 
Now, let $G' = \{ (A,\textbf{a})\times (B,\textbf{b})\mid (B,\textbf{b})\in G\}$. 
Then $F\cup G'$ is a frontier for $(A,\textbf{a})$ w.r.t.~all structures.
To see this, note that each structure in $F\cup G'$ homomorphically maps to $(A,\textbf{a})$. Furthermore,
$(A,\textbf{a})$ does not map to any structure in $F$ (by initial assumption) or in $G'$ 
(because $G'$ consists of unsafe structures while $(A,\textbf{a})$ is safe). Finally, 
consider any $(C,\textbf{c})\stricthom (A,\textbf{a})$. If $(C,\textbf{c})$ is safe, then 
it maps to a structure in $F$. Otherwise, it maps to a structure in $G$ and hence also to the corresponding structure in $G'$. In either case, it maps to a structure in $F\cup G'$.
\end{proof}

Putting everything together, we obtain the main result of this section.
We call a CQ $q$ c-acyclic (or acyclic, or c-connected) if the structure
$\widehat{q}$ is c-acyclic (resp.~acyclic, c-connected).

\begin{theorem} \label{thm:characterizations-main}
Fix a schema and fix $k\geq 0$.
\begin{enumerate}
    \item If $\mathcal{C}$ is a class of $k$-ary CQs such
    that $\widehat{C}$ has bounded expansion, then every CQ $q\in\mathcal{C}$ is uniquely characterizable w.r.t.~$\mathcal{C}$ by finitely many positive and negative examples (which can be effectively constructed from the query).
    \item A $k$-ary CQ $q$ is uniquely characterizable by finitely
    many positive and negative examples
    (w.r.t.~the class of all $k$-ary CQs) iff $q$ is
    logically equivalent to a c-acyclic CQ.
    Moreover, for a c-acyclic CQ, a uniquely characterizing set of examples can be constructed in polynomial time. 
    \item Assume $k\geq 1$ 
    and let $\mathcal{C}_{ca}$ be the class of $k$-ary CQs that are c-connected and acyclic. Then every $q\in \mathcal{C}_{ca}$ is uniquely characterizable w.r.t.~$\mathcal{C}_{ca}$ by
    finitely many positive and negative examples belonging to
    $\widehat{\mathcal{C}_{ca}}$. Moreover, the set of examples in question can be constructed in polynomial time.
    \end{enumerate}
\end{theorem}

\begin{remark}\label{rem:unsafe}
For the purpose of applications discussed in Section~\ref{sec:applications}, we note that Theorem~\ref{thm:characterizations-main} remains true if the safety condition for CQs were to be dropped. Indeed, the proof in this case is even simpler, as it does not require Lemma~\ref{lem:characterisation-safe} and Lemma~\ref{lem:frontier-safe}.
\end{remark}

\subsection*{Examples showing that Theorem~\ref{thm:characterizations-main}(3) cannot easily be generalized.}

Theorem~\ref{thm:characterizations-main}(3) applies to CQs of arity $k \geq 1$ that are c-connected, and acyclic. None of these restrictions can be dropped. 
Recall that a CQ of arity zero is called a \emph{Boolean CQ}.


\newcommand{\bG}{\operatorname{{\bf G}}}
\newcommand{\bI}{\operatorname{{\bf I}}}
\newcommand{\bJ}{\operatorname{{\bf J}}}
\newcommand{\bA}{\operatorname{{\bf A}}}
\newcommand{\bB}{\operatorname{{\bf B}}}
\newcommand{\dom}{\operatorname{dom}}
\newcommand{\bT}{\operatorname{{\bf T}}}
\newcommand{\bR}{\operatorname{{\bf R}}}
\newcommand{\bS}{\operatorname{{\bf S}}}
\newcommand{\invd}{\operatorname{sdr}}

\begin{theorem}\label{thm:unary-is-necessary}
The Boolean acyclic connected CQ ~
$\bT() \textrm{ :- } Ry_1y_2\land Ry_2y_3\land Ry_3y_4 \land Ry_4y_5$
is not characterized, w.r.t.~the class of Boolean acyclic connected CQs, by finitely many acyclic positive and negative examples.
\end{theorem}

\begin{proof} 
For the sake of a contradiction, assume that $(E^+,E^-)$ is a finite collection of acyclic
positive and negative examples that uniquely characterizes $\bT$ within the class of Boolean acyclic connected CQs. Let $n$ be a bound on the size of the examples in $E^+$ and $E^-$. 
Consider now the following Boolean acyclic connected  conjunctive query $\bR$ (cf.~Figure~\ref{fig:counterexample}):

\[ \bR() \text{ :- } 
\mathop{\bigwedge_{i\in\{1,\ldots,n\} \text{ odd}}}_{j\in\{1,2,3\}} \hspace{-5mm} R(z_j^i,z_{j+1}^i) ~~~\land
\mathop{\bigwedge_{i\in\{1,\ldots,n\} \text{ even}}}_{j\in\{2,3,4\}} \hspace{-5mm} R(z_j^i,z_{j+1}^i) ~~~\land
\bigwedge_{i\in\{1,\ldots,n\} \text{ even}} \hspace{-5mm} R(z^i_3,z_4^{i-1})\land R(z^i_2,z_3^{i+1})
\]

\begin{figure}
\centering
$\begin{array}{c@{}c@{}c@{}c@{}c@{}c@{}c@{}c@{}c@{}c@{}c}
         &          & z^2_5    &          &          &          & z^4_5    &          &          &  \\ 
         &          & \uparrow &          &          &          & \uparrow &          &          &  \\ 
z^1_4    &          & z^2_4    &          & z^3_4    &          & z^4_4    &          & z^5_4    &  \\ 
\uparrow & \nwarrow & \uparrow &          & \uparrow & \nwarrow & \uparrow &          & \uparrow & \nwarrow \\ 
z^1_3    &          & z^2_3    &          & z^3_3    &          & z^4_3    &          & z^5_3    & & ~~\cdots \\ 
\uparrow &          & \uparrow & \nearrow & \uparrow &          & \uparrow & \nearrow & \uparrow &  \\ 
z^1_2    &          & z^2_2    &          & z^3_2    &          & z^4_2    &          & z^5_2    &  \\ 
\uparrow &          &          &          & \uparrow &          &          &          & \uparrow &  \\ 
z^1_1    &          &          &          & z^3_1    &          &          &          & z^5_1    &  \\ 
\end{array}$
\caption{The (canonical structure of the) conjunctive query $\bR$}
\label{fig:counterexample}
\end{figure}

Note that $\bR$ and $\bT$ are \emph{not} logically equivalent since $\bR$ does not contain any directed path of length $4$. The mapping  that sends  $z^i_j$ to $y_j$ defines a homomorphism from (the canonical structure of) $\bR$ to (the canonical structure of) $\bT$. This implies that $\bT\subseteq \bR$, that is, every positive example for $\bT$ is also a positive example for $\bR$, and hence, in particular, $\bR$ fits all the positive examples in $E^+$.
Since $\bR$ and $\bT$ are not logically equivalent, $\bR$ must therefore disagree with $\bT$ on one of the 
negative examples, that is, some example $A\in E^-$ is a positive example for $\bR$.
Let $h:\widehat{\bR}\to A$ be a witnessing homomorphism.

\bigskip

\textbf{Claim:}
There are two elements $u$ and $v$ in $\widehat{\bR}$ at distance $2$ (distance, here, is measured in the underlying tree of $\widehat{\bR}$) such that $h(u)=h(v)$.

\bigskip

\emph{Proof of claim:}
Since $A$ has at most $n$ elements it follows that $h$ is not injective, so there are two elements $u'$, $v'$ of $\widehat{\bR}$ that are mapped by $h$ to the same element. Now,
consider the unique path $u'=x_1,\dots,x_m=v'$ in (the underlying tree of) $\widehat{\bR}$ connecting $u'$ and $v'$. Then, $h(x_1),\dots, h(x_m)$ is a walk in (the underlying tree of)
$A$. Indeed, it is a closed walk since $h(x_1)=h(x_m)$. Now, since $h(x_1),\dots,h(x_m)$ is a closed walk in a tree, it must backtrack in some vertex $h(x_i)$. This means that $h(x_{i-1})=h(x_{i+1})$.
\emph{End of proof of claim.}

\bigskip

Let $u=z^i_j$ and $v=z^{i'}_{j'}$ be the two elements at distance $2$ that are mapped by $h$ to the same element in $A$. It follows from  the definition of $\bR$ that, since $u$ and $v$ are at distance $2$, $|i-i'|\leq 1$. We first consider the case where $i=i'$. In this case, we can assume wlog. that  $j'=j+2$. It then follows that $h(z^i_j),h(z^i_{j+1}),h(z^i_{j+2})$ is a directed cycle in $A$, a contradiction because $A\in E^-$,
and $E^-$ was assumed to consist of acyclic structures.
Next, consider the case where $i\neq i'$. Then $|i-i'|=1$. We can assume without loss of generality that $i$ is even and $i'$ is odd. Again it follows directly by the construction of $\bR$ that $j'=j$. Note that $h(z^{i'}_1),\dots,h(z^{i'}_{j'})$ is a directed path in $A$ of length $j'-1$. Similarly $h(z^i_j),\dots,h(z^i_5)$ is a directed path in $A$ of length $5-j$. Since $h(z^i_j)=h(z^{i'}_{j'})$ it follows that we can concatenate the two directed paths obtaining a directed path of length $4+j'-j\geq 4$. Consequently $\widehat{\bT}$ homomorphically maps to $A$, a contradiction because $A\in E^-$.
\end{proof}

This shows that in Theorem~\ref{thm:characterizations-main}(3), 
the restriction to \emph{non-Boolean} queries cannot be dropped. Similarly, 
the restriction to \emph{c-connected} queries cannot be dropped, and
acyclicity cannot be replaced by the weaker condition of c-acyclicity.
\begin{theorem}\label{thm:connectedness-is-necessary}
The unary acyclic CQ ~~
$\bT'(x) \text{ :- } P(x)\land Ry_1 y_2\land Ry_2y_3\land Ry_3y_4 \land Ry_4 y_5$ is not uniquely characterizable, w.r.t.~the class of unary acyclic CQs,  by finitely many acyclic positive and negative examples.
\end{theorem}

\begin{proof} 
Assume for the sake of a contradiction that there are finitely many acyclic positive and negative examples $(E^+,E^-)$ that uniquely characterizes $\bT'(x)$ w.r.t.~the class of unary acyclic CQs. 
We will construct acyclic positive and negative examples $(E'^+,E'^-)$ that uniquely characterizes $\bT$ w.r.t.~the class of Boolean acyclic c-connected CQs,  contradicting Theorem~\ref{thm:unary-is-necessary}.
For each example in $(A,a)\in E^+\cup E^-$, we take every connected component of $A$
(without any distinguished element) and add it to $E'^+$ or $E'^-$,
depending on whether the component in question satisfies $\bT$.
Note that each of these examples is acyclic. 
By construction, the query $\bT$ fits $(E'^+,E'^-)$. We claim that $(E'^+,E'^-)$ uniquely characterizes $\bT$ w.r.t.~the class of Boolean acyclic connected CQs. 
To see this, let $q$ be any Boolean acyclic connected query that fits $(E'^+,E'^-)$.
Now, let $q'(x)$ be the unary acyclic (disconnected) query that is the conjunction of  $q$ with $P(x)$. Then it is not hard to see that $q'(x)$ fits $(E^+,E^-)$, and therefore, $q'(x)$ is equivalent to $\bT'(x)$. From this, it easily follows that $q$ must be equivalent to $\bT$: any counterexample to the equivalence of $q$ and $\bT$ can be extended to a counterexample to the equivalence of $q'(x)$ and $\bT'(x)$ simply by adding an isolated element satisfying $P$. 
\end{proof}

\begin{theorem}\label{thm:acyclicity-is-necessary}
The unary c-acyclic c-connected CQ ~
$\bT''(x) \text{ :- } Ry_1 y_2\land Ry_2y_3\land Ry_3y_4 \land Ry_4 y_5\land \bigwedge_{i=1\ldots 5} Rxy_i$
~
is not uniquely characterizable, 
w.r.t.~the class of unary c-acyclic c-connected CQs, by 
finitely many c-acyclic positive and negative examples.
\end{theorem}

\begin{proof} 
Assume towards a contradiction that $\bT''$ is uniquely characterized, w.r.t.~the class of unary c-acyclic c-connected CQs, by
a finite collection of c-acyclic positive and negative examples $(E^+,E^-)$.
Let $E'^-$ be the set that contains, for each $(A,a)\in E^-$, the substructure $A'$
of $A$ consisting of all elements $b$ satisfying $R(a,b)$. Note that this substructure
does not include $a$ itself (because if $R(a,a)$ was true in $A$, then 
$(A,a)$ would have been a positive example for $\bT''$) and therefore must be acyclic. 
We claim that $(\{\widehat{\bT}\}, E'^-)$ uniquely characterizes $\bT$ w.r.t.~the class of
Boolean acyclic connected CQs,
contradicting Theorem~\ref{thm:unary-is-necessary}. 

Let $q$ be any Boolean acyclic connected conjunctive query that fits $(\{\widehat{\bT}\}, E'^-)$. Let $q'(x)$ be the conjunctive query expressing that $q$ holds true in the 
substructure consisting of elements reachable by an $R$-edge from $x$. Note that $q'(x)$
can be obtained by extending $q$ with an additional conjunct $R(x,y)$ for every variable $y$ occurring in $q$. Also, note that $q'(x)$ is c-acyclic and c-connected.
Since $\widehat{\bT}$ is a positive example for $q$, 
there is a homomorphism $h:\widehat{q}\to\widehat{\bT}$. By extending $h$ in the obvious way, we have that $\widehat{q'}\to \widehat{\bT''}$. Therefore, 
$\bT''\subseteq q'$, and hence 
every positive example
for $\bT''$ is also a positive example for $q'$, and hence $q'$ fits all the positive
examples in $E^+$. Similarly, $q'$ fits all negative examples in $E^-$, because, if it did not,
then there would be a homomorphism from $\widehat{q'}$ to some $(A,a)\in E^-$, 
from which it would clearly follow that $q$ homomorphically maps to the corresponding $A'\in E'^-$,
which we know is not the case. Therefore, $q'$ fits all examples in $(E^+,E^-)$, and hence,
$q'$ is logically equivalent to $\bT''$. It follows that $q$ must be logically
equivalent to $\bT$: any counterexample for the equivalence of $q$ and $\bT$ can be extended to a counterexample for the equivalence of $q'$ and $\bT''$ by adding a fresh distinguished element connected to all existing elements by means of an $R$-edge.
\end{proof}

\section{Exact learnability with membership queries}
\label{sec:learning}

The unique characterization results in the previous section immediately imply (not-necessarily-efficient) exact learnability results:

\begin{theorem} \label{th:learn-bounded-expansion}
Fix a schema and $k\geq 0$.
Let $\mathcal{C}$ be a computably enumerable 
class of $k$-ary CQs. If
 $\widehat{\mathcal{C}}$ has bounded expansion,
then
$\mathcal{C}$ is exactly learnable with membership queries.
\end{theorem}

The learning algorithm in question simply enumerates all queries $q\in\mathcal{C}$ and uses membership queries to test if the goal query fits the uniquely
characterizing set of examples of $q$ (cf.~Theorem~\ref{thm:characterizations-main}(1)).
Unfortunately, this learning algorithm does not run in polynomial time. Indeed, the number of membership queries is not known to be bounded by any fixed tower of exponentials (even for classes $\mathcal{C}$ for which membership can be tested in polynomial time).
For the special case of c-acyclic queries, we can do a little better by taking advantage of the fact that a uniquely characterizing set of examples can be constructed in polynomial time. Indeed, the class of c-acyclic $k$-ary CQs is exponential-time exactly learnable with membership queries: the learner can simply enumerates all c-acyclic queries in order of increasing size. For each query $q$ (starting with the smallest query), it uses Theorem~\ref{thm:characterizations-main}(2) to test, using polynomially many membership queries, whether the goal query is equivalent to $q$. After at most $2^{O(n)}$ many attempts (where $n$ is the size of the goal query), the algorithm is guaranteed to find a query that is equivalent to the goal query.%
\footnote{Similarly, by  Theorem~\ref{thm:characterizations-main}(3), the 
class of unary acyclic c-connected queries is exponential-time exactly learnable with subset queries, where a \emph{subset query}
is an oracle query asking whether a given CQ from the concept class is implied by the goal query. Subset queries correspond precisely to membership queries where the example is the canonical structure of a query from the concept class.}
Our main result in this section improves on this by establishing \emph{efficient} (i.e., polynomial-time) exact learnability:

\begin{theorem}
\label{th:membership} For each schema and $k\geq 0$, the class of c-acyclic $k$-ary CQs
  is efficiently exactly learnable with membership queries.  
\end{theorem}

At a high level, the learning algorithm works by maintaining a c-acyclic hypothesis that is an over-approximation of the actual goal query. At each iteration, the hypothesis is strengthened by replacing it with one of the elements of its frontier, a process that is shown to terminate and yield a query that is logically equivalent to the goal query.
Note, however, that the frontier of a c-acyclic structure does not, in general, consist of c-acyclic structures. At the heart of the proof of Theorem~\ref{th:membership} lies a non-trivial argument showing how to turn an arbitrary hypothesis into a c-acyclic one with polynomially many membership queries. The detailed proof is given in Section~\ref{sec:learning-proofs} below.


The class of \emph{all} $k$-ary queries is not exactly learnable with membership queries (even with unbounded amount of time and the ability to ask an unbounded number of oracle queries), because exact learnability with membership queries would imply that every query in the class is uniquely characterizable, which we know is not the case. On the other hand, we have:

\begin{theorem}[from~\cite{CateDK13:learning}]\label{thm:memb-equiv}
For each schema $\mathcal{S}$ and $k\geq 0$, 
the class of all $k$-ary CQs over $\mathcal{S}$
  is efficiently exactly learnable with membership and equivalence queries.
\end{theorem}

In fact, it follows from results in~\cite{CateDK13:learning} that the larger class of
all \emph{unions of conjunctive queries} is efficiently exactly learnable with 
membership and equivalence queries (for fixed $k$ and fixed schema). Efficient exact learnability with membership and equivalence queries is not 
a monotone property of concept classes, but
 the result from~\cite{CateDK13:learning}
transfers to CQs as well.
For the sake of completeness, a self-contained
proof of Theorem~\ref{thm:memb-equiv} is given below as well.

\begin{remark}\label{rem:learning-unsafe}
For the purpose of applications discussed in 
Section~\ref{sec:applications}, 
we note that Theorems~\ref{th:learn-bounded-expansion}--\ref{thm:memb-equiv} remain true if the safety
condition for CQs were to be dropped.
\end{remark}

\subsubsection*{Related Work}
 There has been considerable prior work that formally studies  the task of identifying some unknown goal query $Q$ from examples. Work in this direction includes learning CQs, Xpath queries, Sparql, tree patterns, description logic concepts, ontologies, and schema mappings among others \cite{BonifatiCL15,CateDK13:learning,GottlobS10,StaworkoW12}. We shall describe mostly the previous work regarding learning CQs. Some of the work in this direction (\cite{Barcelo017,CateD15,Cohen94a,Hirata00,Willard10} for example) assumes that a background structure $A$ is fixed and known by the algorithm. In this setting, an {\em example} is a $k$-ary tuple $(a_1,\dots,a_k)$ of elements in $A$, labelled positively or negatively depending on whether it belongs or not to $Q(A)$. In the present paper (as in \cite{CateDK13:learning,Haussler89}) we do not fix any background structure (i.e., examples are pairs of the form $(A,\textbf{a})$). Our setting corresponds also to the extended instances with empty background in \cite{Cohen95}. 

In both cases a number of different learning protocols has been considered. In the reverse-engineering problem (as defined in~\cite{WeissC17}) it is only required that the algorithm produces a query consistent with the examples. In a similar direction, the problem of determining whether such a query exists has been intensively studied under some variants (satisfiability, query-by-example, definability, inverse satisfiability) \cite{Barcelo017,CateD15,Willard10}. In some scenarios, it is desirable that  the query produced by the learner not only explains the examples received during the training phase, but also has predictive power. In particular, the model considered in \cite{BonifatiCS16} follows the paradigm of identification in the limit by Gold and requires that, additionally, there exists a finite set of examples that uniquely determines the target query $Q$. In a different direction, the model introduced in \cite{GottlobS10}, inspired by the minimum description length principle, requires to produce a hypothesis consistent after some repairs. A third line of work  (see \cite{Cohen93,Haussler89,Hirata00}) studies this problem under Valiant's probably approximately correct (PAC) model. The present paper is part of a fourth direction based on the exact model of query identification by Angluin. In this model, instead of receiving labelled examples, the learner obtains information about the target query by mean of calls to an oracle. As far as we know, we
are the first to study the exact learnability of CQs using a membership oracle. 


\subsection{Proofs for Theorem~\ref{th:membership} and Theorem~\ref{thm:memb-equiv}}
\label{sec:learning-proofs}

To warm up, we first establish Theorem~\ref{thm:memb-equiv}, because its proof is simpler.
The proof relies on the following lemmas. Recall that we denote by $\widehat{q}$ the canonical 
structure of a conjunctive query $q$.

\begin{lemma} \label{lem:minimize-membership}
Let $(A,\textbf{a})$ be any structure such that $\widehat{q_{goal}}\to (A,\textbf{a})$, where
$q_{goal}$ is the goal conjunctive query. Then, using membership oracle queries, we can compute in polynomial 
time (in the size of $(A,\textbf{a})$) a substructure $(A',\textbf{a})$ of $(A,\textbf{a})$ such that $\widehat{q_{goal}}\to (A',\textbf{a})$, and such that 
 $\widehat{q_{goal}}$ does not homomorphically map to any strict substructure of $(A',\textbf{a})$.
\end{lemma}

\begin{proof}
It suffices to iteratively remove one of the facts from the structure, 
and use a membership query to test if
$\widehat{q_{goal}}$ still admits a homomorphism to the structure after removing the fact in question. Once
no further fact can be removed, we have arrived at $(A',\textbf{a})$.
\end{proof}

Recall that we call a structure \emph{safe} if every distinguished element occurs in a fact, that is, the structure is the canonical structure of a conjunctive query.
The proof of the next lemma is obvious.

\begin{lemma}
\label{lem:safety-preservation}
Let $q$ be a $k$-ary conjunctive query over a schema $\mathcal{S}$, and 
let $(B,\textbf{b})$ be any structure over schema $\mathcal{S}$ with $k$ distinguished
elements. If \, $\widehat{q}\to (B,\textbf{b})$, then 
$(B,\textbf{b})$ is safe.
\end{lemma}

We now present the proof of Theorem~\ref{thm:memb-equiv}.

\begin{proof}[Proof of Theorem~\ref{thm:memb-equiv}]
The learning algorithm maintains a structure $(H,\textbf{h})$, which we can intuitively think of as (the canonical structure of) the algorithm's guess of the goal query. The
structure $(H,\textbf{h})$ is refined in a series of iterations in such a way that at each iteration $i$, its value, $(H_i,\textbf{h}_i)$, satisfies the following properties: (i) 
$\widehat{q_{goal}}\to (H_i,\textbf{h}_i)$, and (ii) the size of $(H_i,\textbf{h}_i)$ is bounded by the size of the $q_{goal}$.

We start by considering $(A,\textbf{a})$ where 
$A$ is the structure containing a single node $a$ that satisfies all possible facts over the schema and $\textbf{a}$ is the $k$-ary tuple $(a,\dots,a)$ containing only element $a$.
Clearly, $\widehat{q_{goal}}$ homomorphically maps to this structure.
We apply Lemma~\ref{lem:minimize-membership} to find a minimal
substructure of it into which the goal query maps. We will denote it be $(H_0,\textbf{h}_0)$.

Next, at each stage we perform an equivalence oracle query to test if 
the canonical conjunctive query of $(H_i,\textbf{h}_i)$ is logically equivalent to $q_{goal}$. 
Note that, by Lemma~\ref{lem:safety-preservation}, the structure $(H_i,\textbf{h}_i)$
indeed has a canonical conjunctive query.
If the answer to the equivalence oracle query is ``yes'', then we are done. Otherwise, 
we receive a counterexample $(A,\textbf{a})$. This counterexample must be a structure 
in which
the goal query is true but the hypothesis is false. Thus, we have
$\widehat{q_{goal}} \to (A,\textbf{a})$ and $(H_i,\textbf{h}_i)\not\to (A,\textbf{a})$. 
Recall that we also have $\widehat{q_{goal}}\to (H_i,\textbf{h}_i)$.
It follows that $\widehat{q_{goal}} \to (A,\textbf{a}) \times (H_i,\textbf{h}_i)$.
We now set $(H_{i+1},\textbf{h}_{i+1})$ to be a minimal substructure of
$(A,\textbf{a}) \times (H_i,\textbf{h}_i)$ into which $\widehat{q_{goal}}$ maps (using
Lemma~\ref{lem:minimize-membership} again). It is clear from the construction that
$(H_{i+1},\textbf{h}_{i+1}) \stricthom (H_i,\textbf{h}_i)$, and that the size
of $(H_{i+1},\textbf{h}_{i+1})$ is bounded by the size of $\widehat{q_{goal}}$
(otherwise $\widehat{q_{goal}}$ would not be a minimal substructure of
$(A,\textbf{a}) \times (H_i,\textbf{h}_i)$ into which $\widehat{q_{goal}}$ maps).

All that remains to be shown is that this algorithm terminates after polynomially
many iterations. We show that with each iteration, the domain size of the structure $(H_i,\textbf{h}_i)$ strictly increases. Suppose that the domain size of $(H_i,\textbf{h}_i)$
and $(H_{i+1},\textbf{h}_{i+1})$ is the same. 
We know that the homomorphism (natural projection)  
$h:(H_{i+1},\textbf{h}_{i+1})\to (H_i,\textbf{h}_i)$ is 
surjective, and that every fact of $H_i$ is the $h$-image of a fact of $H_{i+1}$
(otherwise the composition with the homomorphism 
from $\widehat{q_{goal}}$ to $(H_{i+1},\textbf{h}_{i+1})$ would constitute a non-surjective
homomorphism from $\widehat{q_{goal}}$ to $(H_i,\textbf{h}_i)$, which would contradict the 
minimality of $(H_i,\textbf{h}_i)$). Therefore, it cannot also be injective,
otherwise it would be an isomorphism. Therefore, the number of elements of 
$(H_{i+1},\textbf{h}_{i+1})$ is strictly greater than the number of elements in $(H_i,\textbf{h}_i)$. Since the size of each $(H_i,\textbf{h}_i)$ is bounded by the 
size of $q_{goal}$, this shows that the algorithm terminates after at most $n$ 
rounds, where $n$ is the size of $q_{goal}$.

Incidentally, 
note that this algorithm runs in polynomial time (in the size of $q_{goal}$), even 
if the schema and the arity $k$ of the conjunctive query are not fixed but treated as part of the input.
\end{proof}

\newcommand{\ba}{{\bf a}}
\newcommand{\bb}{{\bf b}}
\newcommand{\bc}{{\bf c}}
\newcommand{\bg}{{\bf g}}
\newcommand{\bh}{{\bf h}}
\newcommand{\bff}{{\bf f}}
\newcommand{\CC}{\operatorname{CC}}

Next, in the remainder of this section, we prove Theorem~\ref{th:membership}.

First, we argue that we may restrict attention to schemas consisting of binary relations only. Let $q_{goal}$ be any c-acyclic goal conjunctive query 
 over an arbitrary schema $\mathcal{S}$ and consider the corresponding conjunctive
 query $q_{goal}^*$
over the binary schema $\mathcal{S}^*$, as in Lemma~\ref{lem:binary-reduction}. By Lemma~\ref{lem:binary-reduction}(2),
every membership query w.r.t.~the goal query $q_{goal}^*$ can be efficiently reduced to a
membership query w.r.t.~$q_{goal}$. Therefore, if 
$q_{goal}^*$ can be efficiently identified using membership queries for $q_{goal}^*$,
then $q_{goal}^*$ can also be efficiently identified using membership queries for $q_{goal}$,
and consequently 
also $q_{goal}$ can be efficiently identified using membership queries for $q_{goal}$ (namely, by first computing a conjunctive query $q$ over $\mathcal{S}^*$ that is logically equivalent to $q_{goal}^*$, and then returning $q_*$ as the final answer
(note that, by~Lemma~\ref{lem:binary-reduction}(3), $q_*$ must then be logically equivalent to $q_{goal}$). 
Finally, it is easy to see that $q_{goal}^*$ is 
c-acyclic if and only if $q_{goal}$ is c-acyclic. Therefore, it suffices to prove 
Theorem~\ref{th:membership} for the special case of schemas consisting of binary relations. In the remainder of this section, we will therefore restrict ourselves to binary relations.

\begin{proposition}[Any positive example can be transformed into a c-acyclic one]
\label{pr:toCC}
Let $q_{goal}$ be a c-acyclic goal query.
Given a structure $(A,\ba)$ satisfying $\widehat{q_{goal}}\to(A,\ba)$, using membership queries, we can construct, in time polynomial in $size(A)+size(q_{goal})$, a structure $(B,\bb)$, denoted $\CC(A,\ba)$,  such that
\begin{enumerate}
\item $(B,\bb)$ is c-acyclic
\item $(B,\bb)\to (A,\ba)$
\item $\widehat{q_{goal}}\to (B,\bb)$
\item $\widehat{q_{goal}}$ is not homomorphic to any structure obtained removing some fact of $(B,\bb)$
\end{enumerate}
\end{proposition}

\begin{proof}
We say that a structure is $m$-c-acyclic if every cycle of length at most $m$ goes through a distinguished element. (Here, by a ``cycle'' we mean a cycle in the incidence graph of the structure, and, to simplify the exposition, by the length of the cycle we refer to the number of facts that lie on the cycle).
Note that when this holds for $m=n$, where $n$ is the size of $\widehat{q_{goal}}$, then every homomorphic image of $\widehat{q_{goal}}$ contained in the structure in question must be c-acyclic. 
We will describe a method (using membership queries) that takes a $m$-c-acyclic structure $(A,\ba)$ with $\widehat{q_{goal}}\to (A,\ba)$ and turns it into a $(m+1)$-c-acyclic structure $(A',\ba)$ with $\widehat{q_{goal}}\to (A',\ba)$ and $(A',\ba)\to (A,\ba)$. 
By applying this method repeatedly for increasing $m$ (and always minimizing w.r.t.~$\widehat{q_{goal}}$ using membership queries, cf.~Lemma~\ref{lem:minimize-membership}), we are guaranteed to
reach the situation where we have a structure that is $n$-c-acyclic and therefore, is in fact, c-acyclic. 

Let $(A,\ba)$ be $m$-c-acyclic with $\widehat{q_{goal}}\to (A,\ba)$.
First, we use membership queries to minimize $(A,\ba)$ (cf.~Lemma~\ref{lem:minimize-membership}) and ensure that its size is at most $n$. Next, we say that an edge is \emph{bad} if it is part of a cycle of length $m+1$ that does not contain a distinguished element, and \emph{good} otherwise. 
If there are no bad edges then we are done. Otherwise, let $e = R(c,d)$ be a bad edge. Let $A_1, A_2$ be isomorphic copies of $A\setminus\{e\}$ that are disjoint except for the distinguished elements. Now, let $(B,\textbf{a})$ be the structure obtained by extending the fg-disjoint union
$(A_1,\textbf{a})\uplus(A_2,\textbf{a})$ with additional ``special edges'' $R(c_1, d_2)$ and $R(c_2, d_1)$. 
Clearly, $(B,\ba)\to (A,\ba)$.

\textbf{Claim 1:} for each good edge of $(A,\ba)$, its isomorphic copies belonging to $B$ are good in $(B,\ba)$.

\textbf{Claim 2:} the special edges $R(c_1, d_2)$ and $R(c_2, d_1)$ are both good in $B$. 

Claim 1 is obvious from the construction of $B$, as no new short cycles are introduced.
To see that Claim 2 holds, consider any minimal cycle in $B$ that does not contain any 
distinguished element, and that goes through one of these edges. Then, clearly, 
the cycle must go through both of these edges. That is, it must be of the form 
\[c_1 \xrightarrow{R(c_1, d_2)} d_2 \xrightarrow{~~~\pi~~~} c_2 \xrightarrow{R(c_2, d_1)} d_1 \xrightarrow{~~~\pi'~~~} c_1
\] where $\pi$ is a path contained in $A_2$ and 
$\pi'$ is a path contained in $A_1$. Now, we know that the paths $\pi$ and $\pi'$ must have length at least $m$ (because otherwise $(A,\ba)$ would not be $m$-c-acyclic). Therefore, the entire cycle must have length at least
$2m+2$.

\textbf{Claim 3:} $\widehat{q_{goal}}\to (B,\ba)$.

Claim 3 is essentially proved by an induction on the tree structure of $\widehat{q_{goal}}$, after removing all distinguished elements. More precisely, 
let $h:\widehat{q_{goal}}\to (A,\ba)$. Let $G$ be the substructure of $\widehat{q_{goal}}$ obtained by removing all distinguished elements and facts involving distinguished elements. 
Clearly, $G$ is acyclic, i.e., $G$ can be oriented as a forest. 
By induction on this forest, we can construct a homomorphism
$h':G\to B$, with the additional property that, 
for each element $g$ of $G$, $h'(g)$ is equal to either
$h(g)_1$ or $h(g)_2$. 
Recall that $h(g)_1$ is the copy of $h(g)$ in $A_1$ and that $h(g)_2$ is the copy of $h(g)$ in $A_2$.
Next, let $h''$ be the extension of $h'$ to $\widehat{q_{goal}}$ that agrees with $h$ on all
distinguished elements. We can show that $h'':\widehat{q_{goal}}\to (B,\ba)$. To see this, consider any fact of $\widehat{q_{goal}}$. If that fact involves at least one distinguished element, then the $h$-image of that
fact involves at least one distinguished element of $A$. Therefore, by construction, all the facts obtainable by replacing every non-distinguished element  (if there is such) by one of its two copies, are present in $B$, and therefore, no matter
how $h'$ acts on the non-distinguished element of that fact, the $h''$-image will be present in $B$. Next, consider the case where the fact doesn't involve any distinguished element of $A$. In this case, the fact belongs to $G$ and hence we know that the 
$h'$-image of that fact (which is also the $h''$-image) belongs to $B$. 
This concludes the proof of claim 3.

Next, let $(B',\ba)$ be a minimal substructure of $(B,\ba)$ into which $\widehat{q_{goal}}$ homomorphically maps (obtained using membership queries, cf.~Lemma~\ref{lem:minimize-membership}). 
Clearly, Claim 1-3 above are preserved under the passage from $(B,\ba)$ to its substructure $(B',\ba)$.
We claim that (1) $B'$ must contain at least one of the edges
$R(c_1, d_2)$ and $R(c_2, d_1)$, and (2) for every edge of $A\setminus\{e\}$,
$B'$ must contain at least one of its two isomorphic copies. Because if not, then 
the homomorphism from $\widehat{q_{goal}}$ to $(B',\ba)$ could be composed with the natural projection from $(B',\ba)$ to $(A,\ba)$ to obtain a homomorphism from $\widehat{q_{goal}}$ to a proper substructure of $(A,\ba)$, a contradiction with the assumed minimality of $(A,\ba)$. 
Combined with Claim 1 and 2, this allows us to conclude that $(B',\ba)$ has strictly more good edges than $(A,\ba)$. 
Since the number of edges (good or bad) is bounded by
$n$, by repeating the above procedure, we obtain a structure that has no bad edges,
and therefore, is $(m+1)$-c-acyclic. 
\end{proof}

\begin{proof}[Proof of Theorem~\ref{th:membership}:]
The algorithm maintains a structure, denoted $(H,\bh)$, which can be interpreted as 
(the canonical structure of)
the algorithm's guess of the goal query. At every moment in the execution of the algorithm, $(H,\bh)$, satisfies the following properties:
\begin{enumerate}
\item $\widehat{q_{goal}}\to (H,\bh)$
\item $\widehat{q_{goal}}$ does not homomorphically map to any structure obtained by removing a fact from $(H,\bh)$
\item $(H,\bh)$ is c-acyclic
\end{enumerate}
Note that conditions (1) and (2) imply that $H$ cannot have more elements than $\widehat{q_{goal}}$ and that $size(H)\leq size(\widehat{q_{goal}})$

Initially, $(H,\bh)$ is defined to be $\CC(A,\ba)$ where $\CC(\cdot)$ is defined as in Proposition~\ref{pr:toCC} and 
$A$ is the structure containing a single node $a$ that satisfies all possible facts over the schema and $\ba$ is the $k$-ary tuple $(a,\dots,a)$ containing only element $a$. The algorithm refines $(H,\bh)$ in a sequence of iterations. At each iteration, the algorithm first constructs  the frontier ${\mathcal F}$ of $(H,\bh)$. Note that by condition (3) $(H,\bh)$ is $c$-acyclic. 
Hence, by Theorem~\ref{thm:cacyclic-polyfrontier}, ${\mathcal F}$ can be computed in time polynomial in $size(H)$ (and, hence, in $size(q_{goal})$). Then, the algorithm asks a membership query for each $(F,\bff)$ in $\mathcal F$ until either it receives a 'yes' answer or, otherwise, it exhausts all structures in ${\mathcal F}$ without receiving a 'yes' answer. In the latter case the algorithm stops and returns the canonical query of $(H,\bh)$, cf.~Lemma~\ref{lem:safety-preservation}, 
as it is immediate from the fact that $\widehat{q_{goal}}\to (H,\bh)$ and that $\widehat{q_{goal}}$ does not homomorphically map to any structure in ${\mathcal F}$ that $(H,\bh)$ is homomorphically equivalent to $\widehat{q_{goal}}$, and therefore,
the canonical query of $(H,\bh)$ is logically equivalent to $q_{goal}$. In the former case, the algorithm picks any structure $(F,\bff)$ in ${\mathcal F}$ that (when asked as a membership query) produces a 'yes' answer, updates $(H,\bh)$, by setting $(H,\bh):=\CC(F,\bff)$ and starts a new iteration. It follows immediately that $(H,\bh)$ preserves properties (1-3) above. 

To show the correctness of the algorithm it only remains to see that the number of iterations is polynomially bounded in $size(G)$. This follows directly from the following claim: the domain size of $H$ increases at each iteration. To prove the claim, let $(H_i,\bh_i)$ be the value of $(H,\bh)$ at the $i$th iteration, and note that, by design, we have $(H_{i+1},\bh_{i+1})\to (F,\bff)$ where $(F,\bff)$ belongs to the frontier of $(H_i,\bh_i)$. It follows that $(H_{i+1},\bh_{i+1})\to (H_i,\bh_i)$ and $(H_i,\bh_i)\not\to (H_{i+1},\bh_{i+1})$. Let $h$ be the homomorphism from $(H_{i+1},\bh_{i+1})$
to $(H_i,\bh_i)$. It follows from the fact that $\widehat{q_{goal}}\to (H_{i+1},\bh_{i+1})$ and condition (2) that the image of $(H_{i+1},\bh_{i+1})$ according to $h$ is precisely $(H_i,\bh_i)$. 
Therefore, $h$ must be non-injective (otherwise it would be an isomorphism, contradicting the fact that $(H_i,\bh_i)\not\to (H_{i+1},\bh_{i+1})$). Since $h$ is surjective and non-injective, we can conclude that the domain size of $H_{i+1}$ is larger than the domain size of $H_i$. 
\end{proof}

\section{Another type of data examples: input-output examples}
\label{sec:input-output}

In the previous sections, we have focused on \emph{positive and negative data examples}. However, there is also another type of data example that is natural 
to consider, namely a pair $(I,R)$ where $I$ is an input instance (over a schema $\mathcal{S}$ of the query), and $R$ is a $k$-ary relation over the domain of $I$, where $k$ is the arity of the query. A conjunctive query $q$ \emph{fits} 
$(I,R)$ if $R$ is precisely the set of all tuples over the domain of $I$ that 
satisfy $q$.
Input-output examples are analogous to
\emph{universal examples} for schema mappings as studied in~\cite{AlexeCKT2011}, in that they 
capture the complete behavior of the concept (in our case, the conjunctive query) 
on a database instance.

One input-output example, intuitively, captures the same information as a polynomial number of positive and negative data examples (if we treat $k$ as a constant), 
namely all positive data examples $(I,\textbf{a})$ for $\textbf{a}\in R$
and all negative data examples $(I,\textbf{a})$ for $\textbf{a}\in dom(I)^k\setminus R$. It follows that a CQ is uniquely characterizable by a finite set of input-output data examples if and only if it is uniquely characterizable by a finite set of positive and negative data examples. In the case of connected CQs, in fact, a single input-output example suffices:

\begin{theorem}
Fix a schema and $k\geq 0$, and let $C$ be any class of $k$-ary CQs over the given schema. Then the following are equivalent
for all queries $q\in C$:
\begin{enumerate}
    \item $q$ is uniquely characterized w.r.t.~$C$ by a finite collection of positive and negative data examples. 
    \item $q$ is uniquely characterized w.r.t.~$C$ by a finite collection of input-output examples.
\end{enumerate}
And, if $C$ consists of c-connected CQs (or  $k=0$ and $C$ consists of connected CQs):
\begin{enumerate}
    \item[(3)]  $q$ is uniquely characterized w.r.t.~$C$ by a single input-output example.
\end{enumerate}
Moreover, the equivalences are witnessed by polynomial-time transformations in each direction.
\end{theorem}
\begin{proof}
For the direction from (1) to (2): for given $(E^+,E^-)$, let $E^{io} = \{ (I,q(I)) \mid (I,\textbf{a})\in E^+\cup E^-\}$. Note that whenever a query $q'\in C$ fits $E^{io}$, then
it must also fit $(E^+, E^-)$.
For the converse direction, from (2) to (1), the construction was already described above.

The direction (3) to (2) is trivial.

Finally, for the direction from (2) to (3), let $E^{io}=\{(I_1,q(I_1)), \ldots, (I_n,q(I_n))\}$.
Let $I$ be the disjoint union $\biguplus_{i=1\ldots n}(I_i)$. 
For every tuple $\textbf{a}$ from the domain of $I_i$, and for every c-connected query $q'$, $\textbf{a}\in q'(I_i)$ if and only if $\textbf{a}\in q'(I)$. 
It follows that, whenever two $c$-connected queries $q, q'$ agree on their output on $I$ (that is to say,
$q(I)=q'(I)$), then they must also agree on their output on each $I_i$. Therefore, 
if $E^{io}$ uniquely characterizes $q$ w.r.t.~$C$, then also $(I,q(I))$ uniquely characterizes $q$ w.r.t.~$C$. The same argument applies to connected Boolean CQs.

Incidentally, note that the above argument only works for c-connected CQs. For instance, the CQ
$q(x) =  \exists y (P(x)\land Q(y))$ cannot be uniquely characterized by a 
single input-output example $(I,q(I))$, because if $q(I)$ is non-empty, then
$q(I)=q'(I)$ where $q'(x)$ is the query $P(x)$; whereas if $q(I)$ is empty, 
then $q(I) = q''(I)$, where $q''(x)$ is the query $P(x)\land Q(x)$.
\end{proof}

It follows that all results regarding the existence and polynomial-time computability of unique characterizations in Section~\ref{sec:characterizations} 
remain true when considering input-output data examples instead of (or even in addition to) positive and negative data examples.

Similarly, in the exact learning context, we can also consider a different type of oracle query, namely where the algorithm provides the oracle with an input instance $I$
and the oracle responds with an input-output example of the form $(I,R)$ that fits the target CQ.
We could call such oracle queries \emph{input-output queries}. They naturally capture a scenario in which we have black-box access to an executable version of
the target CQ.
For the same reasons discussed above, input-output queries are no more powerful than membership queries, 
since one input-output query can be simulated by a polynomial number of 
membership queries (assuming $k$ is fixed). Therefore, also, all our results on exact learnability remain true when considering input-output queries instead of membership queries.

\section{Further Applications}
\label{sec:applications}

While our main focus in this paper is on unique characterizability and exact learnability for 
CQs, in this section, we explore some implications for other application domains.

\subsection{Characterizability and learnability of LAV schema mappings}

A schema mapping is a high-level declarative specification of the relationships between two database schemas~\cite{Kolaitis05schema}. Two of the most well-studied  schema mapping specification languages are 
\emph{LAV (``Local-as-View'')} and \emph{GAV (``Global-as-View'')} schema mappings. 

In~\cite{AlexeCKT2011}, the authors studied the question of
when a schema mapping can be uniquely characterized by a finite set of data examples. 
Different types of data examples were introduced and studied, namely positive examples, negative examples, and ``universal'' examples.  
In particular, it was shown in~\cite{AlexeCKT2011} that a GAV schema mapping can be 
uniquely characterized by a finite set of positive and negative examples (or, equivalently, by a finite set of universal examples) if and only if the schema mapping in question is logically equivalent to one that is specified using c-acyclic GAV constraints. 

It was shown in~\cite{AlexeCKT2011} that
every LAV schema mapping is uniquely characterized by a finite set of universal examples, and that there are LAV schema mappings that are not uniquely characterized by any finite set of positive and negative examples. In this section, we will consider the question
\emph{which} LAV schema mappings are uniquely characterizable by a finite set of 
positive and negative examples, and how to construct such a set of examples efficiently. 

We will also consider the exact learnability of LAV schema mappings with membership queries. Exact learnability of GAV schema mappings was studied in~\cite{CateDK13:learning}, where it was shown that GAV schema mappings are learnable with membership and equivalence queries (and, subsequently, also in a variant of the PAC model) but is not exactly learnable with membership queries alone or with equivalence queries alone. The exact learning algorithm for GAV schema mappings from~\cite{CateDK13:learning} was further put to use and validated experimentally in~\cite{Kun18active}.
Here, we consider exact learnability of LAV schema mappings
with membership queries.

\begin{definition}
A \emph{LAV (``Local-As-View'') schema mapping} is a triple $M=(S,T,\Sigma)$ where
$S$ and $T$ are disjoint schemas (the ``source schema'' 
and ``target schema''), and $\Sigma$ is a finite set of 
\emph{LAV constraints}, that is,  first-order sentences of the form $\forall\textbf{x}(\alpha(\textbf{x})\to\exists \textbf{y}\phi(\textbf{x},\textbf{y}))$,
where $\alpha(\textbf{x})$ is an atomic formula using a relation from $S$,
and $\phi(\textbf{x},\textbf{y})$ is a conjunction of atomic formulas using relations from $T$. 
\end{definition}

By a \emph{schema-mapping example} we will mean a pair $(I,J)$ where $I$ is a
structure over schema $\mathcal{S}$ without distinguished elements, and $J$ is a
structure over schema $\mathcal{T}$ without distinguished elements. We say
that $(I,J)$ is a \emph{positive example} for a schema mapping $M=(\mathcal{S},\mathcal{T},\Sigma)$ if $(I,J)$, viewed as a single structure over
the joint schema $\mathcal{S}\cup\mathcal{T}$, satisfies all constraints in $\Sigma$,
and we call $(I,J)$ a \emph{negative example} for $M$ otherwise.
Note that schema-mapping examples were called \emph{data examples} in~\cite{AlexeCKT2011}.
Unique characterizations and learnability with membership queries are defined as before. 
In particular, by a \emph{membership query}, in the context of learning LAV schema mappings, we will mean
an oracle query that consists of a schema-mapping example, which the oracle then labels as
positive or negative depending on whether it satisfies the constraints of the goal LAV schema mapping. It is assumed here, that the source and target schemas are fixed and known to the learner.

Given a fixed source schema $\mathcal{S}$, there are only finitely many different possible
left-hand sides $\alpha$ for a LAV constraint, up to renaming of variables. 
Furthermore, if a schema mapping contains two LAV constraints with the same
left-hand side, then they can be combined into a single LAV constraint by 
conjoining
the respective right-hand sides. Since the
right-hand side of a LAV constraint can be thought of as a CQ,
 this means that, intuitively, 
a LAV schema mapping can be thought of as a finite collection of CQs
(one for each possible left-hand side). In the light of this observation, it is no
surprise that questions about the unique characterizability and learnability of 
LAV schema mappings can be reduced to questions about the unique characterizability 
and learnability of CQs. 


\newcommand{\ATOMS}{\textrm{ATOMS}}

Let us capture this observation a little 
more precisely. Let $k$ be the maximum arity of a relation in $\mathcal{S}$, and
let $\ATOMS_{\mathcal{S}}$ be the finite set of all atomic formulas using a relation from $\mathcal{S}$ and variables from $\{z_1, \ldots, z_k\}$. 
Given a LAV schema mapping $M=(\mathcal{S},\mathcal{T},\Sigma)$ and an $\alpha(\textbf{z})\in \ATOMS_{\mathcal{S}}$, 
we denote by $q_{M,\alpha}(\textbf{z})$ the following first-order formula over schema $\mathcal{T}$:
\[\mathop{\bigwedge_{\forall\textbf{x}(\beta(\textbf{x})\to\exists \textbf{y}\phi(\textbf{x},\textbf{y}))\in \Sigma}}_{
\text{$h:\{\textbf{x}\}\to\{\textbf{z}\}$ a function s.t.~$\beta(h(\textbf{x})) = \alpha(\textbf{z})$}} 
\exists \textbf{y}\phi(h(\textbf{x}),\textbf{y})\]
For example, if $M$ consists of the LAV constraints 
$\forall x_1, x_2, x_3 . R(x_1,x_2,x_3)\to S(x_1,x_2,x_3)$ and 
$\forall x_1, x_2 . R(x_1,x_2,x_2)\to \exists y T(x_1,y)$,
and $\alpha(z_1)$ is $R(z_1, z_1, z_1)$, then
$q_{M,\alpha} = S(z_1,z_1,z_1)\land \exists y T(z_1,y)$.
Similarly, for $\alpha'(z_1, z_2, z_3) = R(z_1, z_2, z_3)$ then
$q_{M,\alpha'} = S(z_1,z_2,z_3)$.
Note that $q_{M,\alpha}(\textbf{z})$ can  be equivalently written as a not-necessarily-safe CQ over $\mathcal{T}$
(by pulling the existential quantifies to the front).

\begin{lemma}\label{lem:lav-lemma1}
Let $M=(\mathcal{S},\mathcal{T},\Sigma)$ be any LAV schema mapping, and let
$\alpha(\textbf{z})\in \ATOMS_\mathcal{S}$ have $k$ distinct variables.
For every structure $(A,\textbf{a})$, over schema $\mathcal{T}$ and with $k$ distinguished elements, the following are equivalent:
\begin{enumerate}
    \item  $(A,\textbf{a})$ is a positive data example for $q_{M,\alpha}(\textbf{z})$,
    \item The schema-mapping example
$(I,J)$ is a positive  example for $M$, where $I$ is the structure over $\mathcal{S}$ consisting of the single fact $\alpha(\textbf{a})$, and $J = A$.
\end{enumerate}
\end{lemma}

We omit the proof, which is straightforward (note that the left-hand side of a LAV constraint can have at most one homomorphism to $I$, and the latter can be extended to the right-hand side of the constraint to $J$ iff
the respective conjunct of $q_{M,\alpha}$ is satisfied. Also note that if $(I,J)$ is a 
positive example for a LAV schema mapping $M$ then, so is $(I,J')$ for $J\subseteq J'$).

Intuitively, Lemma~\ref{lem:lav-lemma1} shows that the behavior of $q_{M,\alpha}$ on arbitrary data examples,
is fully determined by the behavior of $M$ on arbitrary schema-mapping examples.
The converse turns out to be true as well, that is, 
the semantics of a LAV schema mapping $M=(\mathcal{S},\mathcal{T},\Sigma)$ is  determined (up to logical equivalence) by its associated queries $q_{M,\alpha}$ for $\alpha\in \ATOMS_{\mathcal{S}}$:

\begin{lemma}\label{lem:lav-lemma2}
Two LAV schema mappings $M_1 = (\mathcal{S},\mathcal{T},\Sigma_1)$, 
$M_2 = (\mathcal{S},\mathcal{T},\Sigma_1)$ are logically equivalent
iff, for every $\alpha(\textbf{z})\in \ATOMS_{\mathcal{S}}$, 
$q_{M_1,\alpha}(\textbf{z})$ and $q_{M_2,\alpha}(\textbf{z})$ are logically equivalent.
\end{lemma}

\begin{proof}
The left-to-right direction follows immediately from the preceding Lemma.
For the right-to-left direction: suppose $M_1$ and $M_2$ are not logically equivalent.
Then they disagree on some schema-mapping example $(I,J)$. 
Without loss of generality, we may assume that $(I,J)$ is a positive example for 
$M_1$ and a negative example for $M_2$. In particular, one of the LAV constraints in
$\Sigma_2$ is false in $(I,J)$. Since the left-hand side of a LAV constraint consists
of a single atom, it follows that, for some fact $R(\textbf{a})$ of $I$, 
the schema-mapping example $(\{R(\textbf{a})\}, J)$ is a negative example for $M_2$. 
Moreover, an easy monotonicity argument shows that 
$(\{R(\textbf{a})\}, J)$ is a positive example for $M_1$. 
Let $\alpha$ be obtained from the fact $R(\textbf{a})$ by replacing each distinct 
element $a_i$ by a corresponding variable $z_i$. It follows from 
Lemma~\ref{lem:lav-lemma1} that $q_{M_1,\alpha}$ and $q_{M_2,\alpha}$ 
disagree on the structure $(J,\textbf{a})$, and are not logically equivalent.
\end{proof}

It follows directly from the above Lemmas that the unique 
characterizability of a LAV schema 
mapping $M$ reduces to the unique characterizability of each query 
$q_{M,\alpha}$:

\begin{lemma}\label{lem:lav-lemma3}
 For all LAV schema mappings $M=(\mathcal{S},\mathcal{T},\Sigma)$,
 the following are equivalent:
 \begin{enumerate}
 \item $M$ is uniquely characterizable by finitely many positive and negative schema-mapping examples (w.r.t.~the class of all LAV schema mappings over $\mathcal{S},\mathcal{T}$).
 \item For each $\alpha(z_1, \ldots, z_k)\in \ATOMS_{\mathcal{S}}$, 
 $q_{M,\alpha}(z_1, \ldots, z_k)$ is uniquely characterizable by finitely many
 positive and negative data examples w.r.t. the class of all $k$-ary not-necessarily-safe CQs over $\mathcal{T}$.
 \end{enumerate}
\end{lemma}

Intuitively, this shows that a LAV schema mapping is uniquely characterizable iff each of its constraints (joined together according to their left-hand side atom) are.
By combining these lemmas with Theorem~\ref{thm:characterizations-main} (cf.~Remark~\ref{rem:unsafe}),
we can link the 
unique characterizability of a LAV schema mapping to the condition of c-acyclicity.
We say that a LAV schema mapping $M$ is c-acyclic if the right-hand side of each
of its LAV constraints is a c-acyclic not-necessarily-safe CQ. Note that, 
in this case, also $q_{M,\alpha}$ is c-acyclic, for each $\alpha\in \ATOMS_{\mathcal{S}}$.

\begin{theorem}\label{thm:lav-characterizations}
Fix a source schema $\mathcal{S}$ and a target schema $\mathcal{T}$.
A LAV schema mapping $M=(\mathcal{S},\mathcal{T},\Sigma)$ is uniquely characterizable by a finite set of positive and negative schema-mapping examples if and only if $M$ is logically equivalent to a c-acyclic LAV schema mapping. Moreover, if $M$ is c-acyclic, then a uniquely characterizing set of positive and negative schema-mapping examples can be constructed in polynomial time (for fixed $\mathcal{S}, \mathcal{T}$).
\end{theorem}

\begin{proof}
The direction going from c-acyclicity to the uniquely characterizing set of schema-mapping examples, follows immediately from the above lemmas together with Theorem~\ref{thm:characterizations-main}. For the other direction, assume 
that $M$ is uniquely characterizable by finitely many positive and negative schema-mapping
examples. It follows by Lemma~\ref{lem:lav-lemma3} that each $q_{M,\alpha}$ is uniquely characterizable by finitely many positive and negative data examples. Hence, 
each $q_{M,\alpha}$ is logically equivalent to a c-acyclic not-necessarily-safe conjunctive query $q'_{M,\alpha}$. Finally, let $M'=(\mathcal{S}, \mathcal{T}, \Sigma')$,
where $\Sigma'$ consists of all LAV constraints of the form 
$\forall \textbf{z}(q_{M,\alpha}(\textbf{z}) \to \alpha(\textbf{z}))$
for $\alpha(\textbf{z})\in \ATOMS_\mathcal{S}$. Then $M'$ is c-acyclic and logically 
equivalent to $M$.
\end{proof}

Similarly, Lemma~\ref{lem:lav-lemma1} and Lemma~\ref{lem:lav-lemma2}, together with Theorem~\ref{th:membership}, directly imply:

\begin{theorem}
Fix a source schema $\mathcal{S}$ and a target schema $\mathcal{T}$.
The class of c-acyclic LAV schema mappings over $\mathcal{S}, \mathcal{T}$ is efficiently exactly learnable with membership queries. 
\end{theorem}

Note that the class of \emph{all} LAV schema mappings over $\mathcal{S}, \mathcal{T}$ is \emph{not} exactly learnable with membership queries (assuming that $\mathcal{S}$ is non-empty and $\mathcal{T}$ contains a relation of arity at least 2). This follows immediately from the existence of LAV schema mappings that are not uniquely
characterizable by finitely many positive and negative schema-mapping examples. 

As mentioned earlier, LAV schema mappings and GAV schema mappings are two of the
most well-studied schema mapping languages. GLAV (``Global-and-Local-As-Views'') schema mappings is another, which forms a common generalization. An important remaining open question in the area of example-driven approaches to schema mapping design is the 
following~\cite{AlexeCKT2011}: \emph{which GLAV schema mappings are uniquely characterizable by a finite set of examples}?

\subsection{Learning description logic concept expressions and ABoxes}

Description logics are formal specification languages used to 
represent domain knowledge. Example-driven and machine-learning based approaches
have a long history in this area, and have received renewed interest in the last years~\cite{Ozaki2020:learning}, in particular, for
ontologies specified in the lightweight description logic $\mathcal{ELI}$, and focusing
on the exact learnability
of ontologies using entailment queries and equivalence queries.
As we show in this section, our results on c-acyclic CQs have some implications
for the exact learnability of $\mathcal{ELI}$ 
concept expressions.



\begin{definition}[$\mathcal{ELI}$ Concept expressions, ABoxes, TBoxes]
Let $N_C, N_R, N_I$ be fixed, disjoint sets,
whose members we will refer to as ``concept names'', ``role names'', and ''individual names'', respectively. The sets $N_C$ and $N_R$ are assumed to be finite,
while $N_I$ is assumed to be infinite.

A \emph{concept expression} $C$ is an expression built up
from from concept names in $N_C$ and $\top$, using
conjunction ($C_1\sqcap C_2$) and existential restriction ($\exists r.C$ or $\exists r^-.C$, where $r\in N_R$). 

An \emph{ABox} is a finite set of \emph{ABox axioms} of the form $P(a)$ and/or $r(a, b)$, where $P\in N_C$, $r\in N_R$, and $a,b\in N_I$. 

A \emph{TBox} is a finite set of \emph{TBox axioms} $C\sqsubseteq D$, where $C, D$ are concept expressions.
\end{definition}

The semantics of these expressions can be explained by translation to first-order logic:

\begin{definition}
The \emph{correspondence schema} is the schema that contains a
unary relation for every $A\in N_C$ and a binary relation for every
$r\in N_R$. Through the standard translation from description logic to first-order logic (cf.~Table~\ref{tab:standard-translation}), every concept expression $C$ translates to a first-order formula $q_C(x)$ over the correspondence schema. By extension, every TBox $\mathcal{T}$ translates to a finite first-order theory $\mathcal{T}^\text{fo}$, where $C_1\sqsubseteq C_2$ translates to 
$\forall x(q_{C_1}(x)\to q_{C_2}(x))$.
\end{definition}

\begin{table}
    \[
    \begin{array}{rcl}
      q_{P}(x) &=& P(x) ~ \text{ for $P\in N_C$} \\
      q_\top(x) &=& \top \\
      q_{C_1\sqcap C_2}(x) &=& q_{C_1}(x) \land q_{C_2}(x) \\
      q_{\exists r.C}(x) &=& \exists y (r(x,y)\land q_{C}(y)) \\
      q_{\exists r^-.C}(x) &=& \exists y (r(y,x)\land q_{C}(y)) \\
    \end{array}
    \]
    \caption{Standard translation from concept expressions to first-order logic}
    \label{tab:standard-translation}
\end{table}

An ABox can equivalently be viewed as a finite structure (without distinguished elements),
whose domain consists of individual names from $N_I$, and whose facts are the ABox assertions. 
Since $N_I$ is assumed to be infinite, every finite structure over the 
correspondence schema can (up to isomorphism) be represented as an ABox. Therefore, in what follows we will use ABoxes and structures interchangeably. 

We can think of an ABox as a (possibly incomplete) list of facts, and a TBox as domain knowledge in the form of rules for deriving more facts. This idea underlies the next definition:

\begin{definition}
A \emph{QA-example} is a pair $(\mathcal{A},a)$ where $\mathcal{A}$ is an 
ABox and $a\in N_I$. We say that $(\mathcal{A},a)$ is a \emph{positive} QA-example for 
a concept expression $C$ relative to a TBox $\mathcal{T}$ 
if $a\in\text{certain}(C,\mathcal{A},\mathcal{T})$ where
 $\text{certain}(C,\mathcal{A},\mathcal{T}) = \bigcap\{ q_C(B) \mid \text{$\mathcal{A}\subseteq B$ and $B\models \mathcal{T}^\text{fo}$}\}$.
If $a\not\in\text{certain}(C,\mathcal{A},\mathcal{T})$, we say that $(\mathcal{A},a)$ is a negative QA-example for 
$C$ relative to $\mathcal{T}$.
\end{definition}

The name \emph{QA-example}, here, reflects the fact that the 
task of computing $\text{certain}(C,\mathcal{A},\mathcal{T})$ is commonly known as \emph{query answering}. 
It is one of the core inference tasks studied in the description logic literature. In general, there are two variants of the definition of $\text{certain}(C,\mathcal{A},\mathcal{T})$: one where $B$ ranges over finite structures, and one where $B$ ranges over all, finite or infinite, structures. The description logic $\mathcal{ELI}$ that we consider here has been shown to be \emph{finitely controllable} \cite{Barany2013querying}, meaning that both definitions are equivalent. For more expressive description logics, this is in general not the case.

\begin{table}
\[\begin{array}{l@{~~~~~}l@{~~~}l}
& \text{Example} & \text{First-order logic translation} \\[1mm]
\text{ABox:} & \mathcal{A} = \{ P(a), r(a,b) \} \\
\text{TBox:} & \mathcal{T} = \{ P \sqsubseteq Q\sqcap \exists r.P\} 
  & \mathcal{T}^\text{fo} = \{\forall x(P(x)\to Q(x)\land \exists y(r(x,y)\land P(y))\} \\
\text{Concept expr:} &  C = \exists r.Q 
  & q_{C}(x) = \exists y(r(x,y)\land Q(y))
\end{array}\]
\caption{Example description logic ABox, TBox and concept expression}
\label{tab:example-dl}
\end{table}

\begin{example}
Consider the ABox, TBox, and concept expression in Table~\ref{tab:example-dl}.
Every model of $\mathcal{T}^\text{fo}$ containing the facts in $\mathcal{A}$ must contain also 
$r(a,c)$ and $Q(c)$ for some $c\in N_I$. It follows that 
$a\in \text{certain}(C,\mathcal{A},\mathcal{T})$. In other words,
 $(\mathcal{A},a)$ is a positive QA-example for $C$ relative to $\mathcal{T}$.
 On the other hand, $(\mathcal{A},b)$ is a negative QA-example for $C$ relative to $\mathcal{T}$.
\end{example}

See~\cite{Baader2017:introduction} for more details on description logic syntax and semantics.
We now explain how our results from Section~\ref{sec:characterizations} and~\ref{sec:learning} can be applied here.
Although a QA-example is just a data example with one distinguished element,
over the correspondence schema, the definition of \emph{positive}/\emph{negative} QA-examples diverges from the definition of positive/negative data examples, because of the TBox $\mathcal{T}$.
For the special case where $\mathcal{T}=\emptyset$, the two
coincide:

\begin{lemma}
\label{lem:qa-examples}
Let $\mathcal{T}=\emptyset$. A QA-example $(\mathcal{A},a)$ is a positive (negative) QA-example for a concept expression $C$ relative to $\mathcal{T}$ iff $(\mathcal{A},a)$ is a positive (negative) data example for $q_C(x)$.
\end{lemma}

Lemma~\ref{lem:qa-examples} follows from the well-known monotonicity property of CQs (i.e., whenever $A\subseteq B$, then $q(A)\subseteq q(B)$), which implies that $\text{certain}(C,\mathcal{A},\emptyset) = q_{C}(\mathcal{A})$.

Concept expressions turn out to correspond precisely to unary,
acyclic, c-connected CQs:

\begin{lemma} \label{lem:dl-acyclic}
The standard translation $q_C(x)$ of every $\mathcal{ELI}$  concept expression $C$ is equivalent to a not-necessarily-safe unary CQ that is acyclic and c-connected.
Conversely, every unary, acyclic, c-connected not-necessarily-safe CQ over the correspondence schema is logically equivalent to $q_C(x)$ for some
$\mathcal{ELI}$  concept expression $C$.
\end{lemma}

 Both directions of Lemma~\ref{lem:dl-acyclic} can be proved using a straightforward induction.

The above two lemmas, together with Theorem~\ref{thm:characterizations-main}(2)
and Theorem~\ref{th:membership} (cf.~Remark~\ref{rem:unsafe} and Remark~\ref{rem:learning-unsafe}) immediately yield our main result here.
We say that a collection of positive and negative QA-examples 
\emph{uniquely characterizes} a concept expression $C$ relative to a 
TBox $\mathcal{T}$ if $C$ fits the examples (relative to $\mathcal{T}$)
and every other concept expression that does so is equivalent (relative to $\mathcal{T}$) to $C$.
By a \emph{QA-membership query} we mean an oracle query consisting of a
QA example, where the oracle answers yes or no depending on whether the
input is a positive QA example or a negative QA example for the goal concept,
relative to the TBox. It is assumed that the TBox is fixed and known 
to the learner.


\begin{theorem}\label{thm:dl-main}
Let $\mathcal{T}=\emptyset$. 
Every $\mathcal{ELI}$ concept expression is uniquely 
characterizable  by a finite collection of positive and negative QA examples (relative to $\mathcal{T}$), which can be computed in polynomial time.
Furthermore, the class of $\mathcal{ELI}$ concept expressions is efficiently exactly learnable with QA-membership queries.
\end{theorem}

Moreover, by Theorem~\ref{thm:characterizations-main}(3), the uniquely characterizing examples can be constructed so that each example $(\mathcal{A},a)$ is the canonical QA-example of a concept expression.
By the \emph{canonical QA-example} of a concept expression $C$, here, we mean the QA-example that (viewed as a structure with one distinguished element) is the canonical structure of the not-necessarily-safe CQ $q_C(x)$.

Theorem~\ref{thm:dl-main} remains true when the concept language is extended with unrestricted existential quantification (of the form $\exists.C$) and a restricted form of the \textbf{I}-\textbf{me} self-reference construct introduced in~\cite{Marx2002narcissists}, namely where the $\textbf{I}$ operator can only occur once, and in the very front of the concept expression. Indeed, it can be shown that this extended concept language (by a straightforward extension of the standard translation) captures precisely the class of c-acyclic unary not-necessarily-safe CQs over the correspondence schema.

This raises the question if Theorem~\ref{thm:dl-main} holds true for arbitrary TBoxes. Since publication of the conference version of this
paper~\cite{tCD2021:conjunctive} (in which we asked the same question), some answers have been obtained.
In~\cite{Funk2021Actively}, it is shown that the answer to this question is \emph{No}
if the TBox is treated as part of the input to the learning algorithm. Indeed, it is 
shown that the problem becomes not efficiently exactly learnable with membership and equivalence
queries. On the other hand, a positive answer is given in~\cite{Funk2021Actively} for a weaker version of the question, namely for the description logic $\mathcal{EL}$, when the learning algorithm is also allowed to ask equivalence queries. In~\cite{funk2022:frontiers}, furthermore,
a positive answer is given for another variant of the above question where the TBox is specified in the description logic DL-Lite, and where the learning algorithm is allowed to ask membership 
and equivalence queries. At the heart of this learning algorithm lies an extension of our frontier construction from Section~\ref{sec:frontier-closed}, also obtained in~\cite{funk2022:frontiers}. 


\section{Acknowledgements}

This paper largely grew out of discussions at Dagstuhl Seminar 19361 (``Logic and Learning'') in Sept.~2019. Victor Dalmau was supported by MICCIN grants TIN2016-76573-C2-1P and PID2019-109137GB-C22. 
Balder ten Cate was supported by the European Union's Horizon 2020 research and innovation programme (MSCA-101031081).
We thank Carsten Lutz, Raoul Koudijs, and Phokion Kolaitis for helpful discussions and for pointing out some flaws in earlier versions of this paper.

\bibliographystyle{plainurl}
\bibliography{bib}

\end{document}